\documentclass{article}
\usepackage{ijcai-format/ijcai25}

\usepackage{times}
\usepackage{soul}
\usepackage{url}
\usepackage[hidelinks]{hyperref}
\usepackage[utf8]{inputenc}
\usepackage[small]{caption}
\usepackage{graphicx}
\usepackage{amsmath}
\usepackage{amsthm}
\usepackage{booktabs}
\usepackage{algorithm}
\usepackage{algorithmic}
\usepackage[switch]{lineno}

\newtheorem{theorem}{Theorem}

\usepackage{amssymb}

\usepackage{subcaption}
\usepackage{multirow}
\usepackage{makecell}
\usepackage[table, dvipsnames]{xcolor}

\usepackage{cleveref}

\newenvironment{proofsketch}[1][Proof Sketch]{%
  \proof[#1]%
}{\endproof}

\newcounter{examplecounter}
\newenvironment{Example}
{\refstepcounter{examplecounter}\par\noindent\textbf{Example \theexamplecounter.}\hspace{1pt}}
    {\hfill\qedsymbol\par}
\newtheorem{lemma}[theorem]{Lemma}

\newtheorem{corollary}[theorem]{Corollary}

\newtheorem{definition}{Definition}

\newcommand{\withoutTwithoutLacp}{CC-ACP}

\newcommand{\withoutTwithLacp}{CCFP-ACP}

\newcommand{\withoutTwithoutLaop}{CC-AOP}

\newcommand{\withoutTwithLaop}{CCFP-AOP}

\newcommand{\withoutTwithoutLdcp}{CC-DCP}

\newcommand{\withoutTwithLdcp}{CCFP-DCP}

\newcommand{\withoutTwithoutLdop}{CC-DOP}

\newcommand{\withoutTwithLdop}{CCFP-DOP}

\usepackage{mathtools}

\definecolor{mygray}{gray}{0.95}
\definecolor{mygray2}{gray}{0.85}

\usepackage[many]{tcolorbox}

\newtcolorbox{boxB}{
    boxrule = 1.5pt,
    rounded corners,
    arc = 5pt   
}

\newtcolorbox{mybox}[2][]{%
attach boxed title to top center
               = {yshift=-8pt},
  colback      = black!5!white,
  colframe     = black!75!black,
  fonttitle    = \bfseries,
  colbacktitle=white!80!white,
  coltitle = gray!10!black,
  title        = #2,#1,
  enhanced,
}

\newtcolorbox[auto counter, number within=section]{myboxA}[2][]{enhanced, 
  attach boxed title to top center={yshift=-3.5mm,yshifttext=-1mm},
  colback=white,
  colframe=black,
  colbacktitle=white,
  fonttitle=\bfseries,
  coltitle=black,
  boxed title style={colframe=white, coltext=black},
  title=#2, #1}

\newtcolorbox{boxC}[3][left=2pt, right=2pt]
{
  colframe = #2!25,
  colback  = #2!5,
  coltitle = #2!20!black,  
  title    = {#3},
  #1,
}

\ifdefined\DRAFT

\newcommand{\eden}[1]{\textcolor{blue}{(Eden says: #1)}}
\newcommand{\hodaya}[1]{\textcolor{red}{(Hodaya says: #1)}}
\else

\newcommand{\eden}[1]{}
\newcommand{\hodaya}[1]{}
\fi

\newtheorem*{rtheorem}{\theremindertheorem}

\newcommand{\theremindertheorem}{}

\newenvironment{reminder}[1]
  {\renewcommand{\theremindertheorem}{Reminder of Theorem #1}\begin{rtheorem}}
  {\end{rtheorem}}

\newtheorem*{rlemma}{\thereminderlemma}

\newcommand{\thereminderlemma}{}

\newenvironment{reminderlemma}[1]
  {\renewcommand{\thereminderlemma}{Reminder of Lemma #1}\begin{rlemma}}
  {\end{rlemma}}

\newcommand{\universe}{P}
\newcommand{\running}{R}

\usepackage{stackengine}

\newcommand{\topvR}{T_{v}(\running)}
\newcommand{\topvRP}[2]{T_{#1}(#2)}

\title{Control for Coalitions in Parliamentary Elections}

\author{
Hodaya Barr
\and
Eden Hartman\and
Yonatan Aumann\And
Sarit Kraus\\
\affiliations
Bar-Ilan University, Israel\\
\emails
odayaben@gmail.com,
eden.r.hartman@gmail.com,
aumann@cs.biu.ac.il,
sarit@cs.biu.ac.il
}

\begin{document}
\crefname{connection}{Connection}{Connections}

\maketitle
\pagestyle{plain}

\begin{abstract}


    The traditional election control problem focuses on the use of control to promote a single candidate. In parliamentary elections, however, the focus shifts: voters care no less about the overall governing coalition than the individual parties' seat count.
    This paper introduces a new problem: controlling parliamentary elections, where the goal extends beyond promoting a single party to influencing the collective seat count of coalitions of parties.

    We focus on plurality rule and control through the addition or deletion of parties. 
    Our analysis reveals that, without restrictions on voters’ preferences, these control problems are W[1]-hard. 
    In some cases, the problems are immune to control, making such efforts ineffective.

    We then study the special case where preferences are symmetric single-peaked.
    We show that in the single-peaked setting, aggregation of voters into types allows for a compact representation of the problem.  Our findings show that for the single-peaked setting, some cases are solvable in polynomial time, while others are NP-hard for the compact representation - but admit a polynomial algorithm for the extensive representation.




    %
\end{abstract}

\section{Introduction}
The problem of control in voting has been extensively studied  (see, \cite{liu2009parameterized,betzler2009parameterized,faliszewski2016control}).
In this problem, an election organizer, or \emph{chair}, aims to promote \emph{one} candidate of choice, by influencing the election structure, e.g. adding or removing other candidates. In parliamentary systems, however, the fate of proposals is often not decided by a single party, but rather by blocs, or \emph{coalitions} of like-minded parties. The same is also true for many governing bodies and committees. In such cases, the determining factor is the aggregate voting power held by the blocs for and against the given proposal, not by a single party. More generally, while citizens - as well as a possible manipulative chair  - may have one \emph{most preferred} party, the success of an entire \emph{bloc of parties}, e.g. right, center, left, is also of prime importance. 
In such settings, the traditional model of \emph{control} is inadequate, as it fails to capture the full range of chair goals.







This paper aims to fill this gap.  
We introduce and study the problem of \emph{control for coalitions of parties}. At the same time, we acknowledge that the chair may also be interested in promoting a specific party.  So, we define a combined control problem, where the chair has a double objective - both promoting an entire coalition and promoting a specific preferred party therein. 

\begin{table*}[h]
\centering
\resizebox{\textwidth}{!}{
\begin{tabular}{|l|>{\columncolor{lightgray}}l|>{\columncolor{mygray2}}l|>{\columncolor{mygray}}l|>{\columncolor{lightgray}}l|>{\columncolor{mygray2}}l|>{\columncolor{mygray}}l|}
\hline
  & \multicolumn{3}{c|}{Control for Coalition (CC)} & \multicolumn{3}{c|}{Control for Coalition and Favored Party (CCFP)} \\ 
\hline
& General 
 & SSP
 & SSP-Con. 
 & General 
 & SSP
 & SSP-Con. 
 \\ 
\hline
 Adding Coalition Parties (ACP) &  W[2]-hard(T\ref{thm:acp-general-C})  & P(T\ref{thm:acp-ssp-C})     & P(T\ref{thm:acp-ssp-C})$^*$  & W[2]-hard(T\ref{thm:acp-general-C})$^*$ &  NP-hard(T\ref{thm:acp-ssp-CFP-non-con}), PP(T\ref{thm:CCFP-ACP-pseudo-polynomial}) & P(T\ref{thm:acp-ssp-CFP-con}) \\
\hline
 Adding Opposition Parties (AOP) & Immune(T\ref{thm:-ccaop-immune})  & Immune(T\ref{thm:-ccaop-immune})  &  Immune(T\ref{thm:-ccaop-immune}) & W[2]-hard(T\ref{thm:aop-general-CFP}) &  NP-hard(T\ref{thm:aop-ssp-CFP-non-con}), PP(T\ref{thm:CCFP-AOP-pseudo-polynomial}) &   P(T\ref{thm:aop-ssp-CFP-con})\\
\hline
 Deleting Coalition Parties (DCP) & Immune(T\ref{thm:-ccdcp-immune}) &  Immune(T\ref{thm:-ccdcp-immune}) & Immune(T\ref{thm:-ccdcp-immune}) &  W[1]-hard(T\ref{thm:dcp-general-CFP}) & NP-hard(T\ref{thm:dcp-ssp-CFP-non-con}), PP(T\ref{thm:CCFP-dcp-pseudo-polynomial}) & P(T\ref{thm:dcp-ssp-CFP-con})\\
\hline
 Deleting Opposition Parties (DOP) & W[1]-hard(T\ref{thm:dop-general-C}) &  P(T\ref{thm:dop-ssp-C}) &P(T\ref{thm:dop-ssp-C})$^*$ & W[1]-hard(T\ref{thm:dop-general-C})$^*$  & NP-hard(T\ref{thm:dop-ssp-CFP-non-con}), PP(T\ref{thm:CCFP-DOP-pseudo-polynomial}) & P(T\ref{thm:dop-ssp-CFP-con})\\
\hline
\end{tabular}
}

\caption{Summary of Results. An asterisk ($^*$) indicates that a theorem follows directly from the referenced theorem, and "PP" is shorthand for a pseudo-polynomial algorithm for the compact representation, which is equivalent to a polynomial algorithm for the extensive representation.}

\label{tbl:sum}

\end{table*}

\paragraph{Contributions.}
Our contribution is twofold. 
First, we introduce and define the problem of control for promoting a coalition, with and without an additional goal of promoting a favored party. Second, we provide a complexity analysis for different settings of the problem, providing algorithmic solutions for some settings, and hardness results for others. We focus on plurality rule and on control through addition or deletion of parties.
The results are summarized in Table \ref{tbl:sum}.

We prove that when there are no restrictions on voters' preferences, in most cases, the problem is W[1]-hard,  meaning that it cannot be solved even by using a Fixed-Parameter Tractability (FPT) algorithm (where the parameter is the number of parties to add or delete). 
In particular, it means that the problems are NP-hard.
In the remaining cases, the problem is immune -- achieving the desired outcome using the specific type of control is not possible.

We then shift our attention to symmetric single-peaked preferences, which have a natural interpretation in the context of political elections. 
We prove that in this case, the number of distinct possible preference orders significantly decreases, enabling the use of a compact input representation that depends only on the number of parties.
%
%
We prove that in the case of symmetric single-peaked preferences, some cases become polynomial-time solvable, while others are NP-hard under the compact representation, but admit a polynomial-time algorithm for the extensive representation.


\paragraph{Organization.}
Section \ref{sec:model} presents the model and required definitions. 
Section \ref{sec:lammas} presents required technical lemmas. Section \ref{sec:adding} focuses on control by adding parties, whereas Section \ref{sec:deleting} focuses on control by deleting parties. 
Section \ref{sec:conclusion} concludes with some future work directions.
Most proofs are deferred to the appendix.

\subsection{Related Work}

This paper extends the problem of constructive control by adding and deleting candidates, originally introduced by \cite{bartholdi1992hard}.
In the original problem, the objective is to promote a single candidate. In contrast, this paper focuses on promoting a coalition -- a set of parties -- and, potentially, a favored party within that set.

\paragraph{Time Complexity.} A common goal is to study the complexity of different types of control and election systems - e.g., \cite{put2016complexity,procaccia2007multi,sina2015adapting,menton2013control,maushagen2020last} (see~\cite{faliszewski2016control} for a comprehensive survey).
Control by adding and deleting candidates under the plurality rule is known to be w[2]-hard problems~\cite{betzler2009parameterized,liu2009parameterized,bartholdi1992hard}.
However, these results do not directly apply to our setting as control in parliamentary elections differs fundamentally. 
While single-winner elections aim to secure the highest number of votes, parliamentary elections focus on maximizing the percentage of votes, regardless of whether this leads to the highest percentage overall.


\paragraph{Symmetric-Single-Peaked Preferences.} The study of single-peaked preferences in elections is quite common (e.g., \cite{elkind2017structured,faliszewski2011multimode,brandt2016handbook}).
The key assumption is that there exists a common ordering of the parties, such that for any voter, their preference between two parties— both located to the left or right of their peak (i.e., most preferred party) — is based on which party is closer to the peak.
In this paper, we focus on \emph{symmetric} single-peaked preferences, as modeled by \cite{masso2011strategy}.
Unlike single-peaked preferences, this model allows voters to compare parties on opposite sides of the peak according to the common ordering as well; it has been explored in various contexts~\cite{bergantinos2015stable,Aziz2022Strategyproof,morimoto2013maximal}.

\paragraph{Strategic Voting and Coalition Preferences.}
Cox~\cite{cox2018portfolio} investigates how voters engage in portfolio-maximizing strategic voting, revealing that their choices are often influenced not just by their preferences for a single party but also by a desire to shape coalition outcomes and enhance governmental influence.
Several studies (e.g., \cite{meffert2010strategic,bowler2010strategic,gschwend2016drives,MCCUEN2010316}) have delved into strategic and tactical coalition voting, showing how voters might cast their votes based on coalition dynamics, even if it means not selecting their top-preferred party.

\section{Model and Definitions}\label{sec:model}
The setting postulates a set of parties,\footnote{Commonly referred to as \emph{candidates} in the context of control problems for single-winner elections.} $\universe = \{p_1, \ldots p_m\}$, that could potentially run in the election, and a set of $n$ voters $V= \{v_1,\ldots,v_n\}$. Each voter, $v$, has a strict preference order, $\succ_v$, over the parties in $\universe$.
An election system is a tuple $(\running, V)$, where $\running \subseteq \universe$ represents the set of parties running in the election. 
In this paper, we assume that the set of voters remains fixed - all voters participating according to their given preference orders. 
However, the chair may influence the set of running parties. 
Thus, for brevity, we represent different election systems by specifying only $\running$.

We denote the top (most preferred) party of voter $v$ in $R$ by $\topvR$ -- that is, $\topvR \in R$ and $\topvR \succ_v p$ for any $p \in R$.

\paragraph{Allocation Rule.}
The goal of parliamentary election is to allocate a fixed number of seats among the running parties. 
An \emph{allocation rule} takes as input the set of running parties $\running$ and the preference orders of the $n$~voters, and determines what fraction of the total number of seats to allocate to each party in $\running$. 
Following~\cite{put2016complexity}, we assume that any fraction is possible, even if the resulting number of seats is not an integer.

In this paper, we focus on the Plurality allocation rule.

\paragraph{Plurality.} 
Given a set of running parties, $\running \subseteq \universe$, each voter contributes its vote for their top party in $R$, $\topvR$.
For each running party $p\in R$, we denote the number of votes $p$ receives by $N(\running,p)$.
The fraction of votes received by party $p$ is denoted by $\pi(\running, p) := N(\running, p)/n$.


\paragraph{Objectives.} 
The chair is an external agent with two goals.

\begin{definition}[Coalition Objective]
    Given a \emph{target-coalition} $C \subset \universe$, and a \emph{coalition target-fraction} $\varphi \in (0,1]$. The chair wishes for the set of running parties $\running$ to satisfy:
    \begin{align}\label{eq:target-coalition-obj}
        \tag{OBJ-C}
        \sum_{p \in C \cap \running} \pi(\running, p) \geq \varphi
    \end{align}
    The objective is to ensure that the fraction of total votes received by the target-coalition is at least $\varphi$.
\end{definition}

\begin{definition}[Favored-Party Objective]
    Given a coalition $C \subseteq \universe$, a \emph{favored-party}, $p_F \in C$, and a \emph{target-favored-party-ratio} $\rho \in [0,1]$. The chair wishes for the set of running parties $\running$ to include $F$ and to satisfy the following:
    \begin{align}\label{eq:target-leader-obj}
        \tag{OBJ-FP}
        \pi(\running, p_F) \geq \rho \cdot \sum_{p \in C \cap P} \pi(\running, p)
    \end{align}
    The objective is to ensure that the fraction of total votes received by the favored party, relative to the votes of the coalition, is at least $\rho$.
\end{definition}

In our runtime analysis, we separately consider two problems: 
\begin{enumerate}
    \item Control for coalition (CC): when the chair cares only about the coalition objective.

    \item Control for coalition and favored-party (CCFP): when the chair cares about both objectives: coalition objective and favored-party Objective.
\end{enumerate}
To highlight the differences, see the example in Appendix~\ref{sec:examples}.




\paragraph{Control by Deleting Parties.} 
Given an integer $1 \leq k \leq m$, the chair is allowed to delete at-most $k$ parties from $\universe$. 

\begin{boxC}{gray}{}
\textbf{Question:} Does there exist a set $K \subseteq \universe$ with $|K| \leq k$ such that the chair's objectives are met when the set of running parties is $\running := \universe \setminus K$?
\end{boxC}

\paragraph{Control by Adding Parties.} 
Here, we have a set $S \subseteq \universe$ of $\ell$ \emph{spoiler} parties that, by default, do not run but can be added by the chair.
The parties $\universe \setminus S$ are called \emph{permanent} -- they run in the election, and the chair cannot change it. 

Given an integer $1 \leq k \leq \ell$, the chair is allowed to add at-most $k$ spoiler parties to run alongside the permanent ones.  

\begin{boxC}{gray}{}
\textbf{Question:} Does there exist a set $K \subseteq S$ with $|K| \leq k$ such that the chair's objectives are met when the set of running parties is $\running := (\universe \setminus S) \cup K$?
\end{boxC}

When $\rho > 0$, we assume that $p_F$ is permanent.

\paragraph{Control Variants.}
We treat coalition and opposition parties separately: the control can be applied either over coalition parties $K \subseteq C$, or over opposition parties $K \subseteq (\universe \setminus C)$.

\subsection{Symmetric Single-Peaked Preferences}
We compare general preferences -- any order over the parties is allowed, to the case of \emph{symmetric single-peaked} (SSP) preferences, defined as follows. 

We assume that the parties in $\universe$ can be positioned on a $[0, 1]$ interval. For simplicity, we assume that the name of each party is its position on the interval.
Based on that, a preference order $\succ_v$ is \emph{symmetric single-peaked} if there exists a point $h(v) \in [0,1]$, called the \emph{peak}, such that the parties are ordered in descending distance from this point.
That is, for any two parties $p_1,p_2 \in \universe$, 
$p_1 \succ_v p_2$, if-and-only-if 
$|h(v) - p_1| < |h(v) - p_2|$.
For brevity, we assume that the name of each voter is its peak.

\paragraph{Problem Size and Complexity Measure.} We prove that there are only $(m^2+1)$ distinct SSP preference orders (see \Cref{app:types} for more details), which we refer to as \emph{types}
\cite{shrot2010agent}.
This enables us to describe the entire set of voters using a \emph{compact} representation – indicating the number of voters of each specific type, which depends only on the number of parties.
This distinction in input size implies that the runtime complexity is not directly comparable to the complexity in the general case, as it is measured relative to a fundamentally different input representation - referred to as \emph{extensive}.

We note that the time complexity of calculating $N(\running, p)$ differs between representations. 
In the compact representation, these calculations depend only on the number of parties, whereas in the extensive representation, they also depend on the total number of voters.



\paragraph{The Target-Coalition.}
Recall that a coalition is a subset of parties $C \subseteq P$. 
For SSP preferences, any coalition can be described as a union of \emph{coalition-intervals}, which are sub-intervals of $[0,1]$ containing only parties in $C$.
We assume that the coalition is given as a set of $q$ disjoint coalition-intervals, $\mathcal{I}^C := \{I^C_1, \ldots, I^C_q\}$, where $q$ is minimal.
If $q > 1$, this implies that there must be at least one opposition party that lies between any two coalition-intervals in $\mathcal{I}^C$.  
When $q=1$, we say that the coalition is \emph{contiguous}.  

We refer to the sub-interval between the coalition-intervals as opposition-interval, denoted by $\mathcal{I}^O := \{I^O_0, \ldots, I^O_{q+1}\}$.


\section{Technical Lemmas}\label{sec:lammas}
We start by proving that adding a party to the set of running parties can only reduce the fraction of votes received by a running party.

\begin{lemma}\label{lemma:adding-only-decreases}
    Let $p_1, p_2 \in P$ and $R \subseteq (P\setminus\{p_2\})$ such that $p_1 \in R$. Then, $\pi(R, p_1) \geq \pi(R \cup \{p_2\}, p_1)$.
\end{lemma}

%



Applying \Cref{lemma:adding-only-decreases} iteratively allows us to conclude:
\begin{corollary}\label{cor:adding-parties-decreases}
    Let $R \subseteq P$, $p \in R$ and $R^+$ such that $R \subset R^+$. Then $\pi(R^+, p) \leq \pi(R, p) $.
\end{corollary}

This claim can also be interpreted as deleting parties from $R^+$; which implies that deleting parties can only increase the votes of a running party.

\subsection{SSP Preferences}
The following lemma proves that a voter's preference between two parties can be determined by their position relative to the average of the parties.

\begin{lemma}\label{lemma:divider-meaning}
    Let $v \in V$ and $p_1 < p_2 \in P$. Then, $p_1 \succ_v p_2$ if-and-only-if $ v < \frac{1}{2}(p_1+p_2)$.
\end{lemma}

%

\Cref{lemma:add_change_to_added_or_not_change} says that adding a party can only shift 
 a voter's top preference to the newly added party or leave the top preference unchanged. 
\Cref{lemma:move-to-nerest_delete} says that deleting a party can only shift a voter's top preference to to one of the parties adjacent to the removed party.
 
\begin{lemma}\label{lemma:add_change_to_added_or_not_change}
    Let $v \in V$, $p_1, p_2 \in P$ and $R \subseteq (P\setminus\{p_2\})$ such that $p_1 \in R$. If $T_v(R) = p_1$, then $T_v(R\cup \{p_2\}) \in \{p_1, p_2\}$.
\end{lemma}

\begin{lemma}\label{lemma:move-to-nerest_delete}
     Let $p_1<\ldots< p_m$ be the common     ordering of the parties on $[0,1]$.
     Let $v \in V$ and $T_v(P) = p_i$. ~~Then, $T_v(P\setminus \{p_i\}) \in \{p_{i-1}, p_{i+1}\}$.
\end{lemma}




\subsubsection{Independent Analysis of Coalition-Intervals}
\Cref{lem:sperate-intervals} proves that modifying the set of parties in one coalition-interval does not influence the set of parties in other coalition-intervals.

\begin{lemma}\label{lem:sperate-intervals}
Let $I^C_1\neq I^C_2 \in \mathcal{I}^C$ be two coalition-intervals, $c_1 \in I^C_1$, $c_2 \in I^C_2$ and $\running \subseteq P\setminus \{c_2\}$ such that $(P\setminus C) \subseteq R$. ~Then: ~~$\pi(\running, c_1) = \pi(\running \cup \{c_2\}, c_1)$.
\end{lemma}

    



Applying the lemma iteratively implies:

\begin{corollary}\label{cor:interval-adding}
    Let $I^C\in \mathcal{I}^C$, $c \in I^C$, and $\running \subseteq P$ such that $c \in R$ and $(P\setminus C) \subseteq R$. ~Then, for any $R^+$ such that $R \subset R^+$ and $R \cap  I^C = R^+ \cap  I^C$: ~~$\pi(\running, c) = \pi(R^+, c)$.
\end{corollary}



\subsubsection{Independent Analysis of opposition-Intervals}


The following lemma proves that each opposition-interval can influence a unique set of voters.
\begin{lemma}\label{lem:sperate-opposition-intervals}
    Let $I^O_1\neq I^O_2$ be two opposition-intervals, ~$o_1 \in I^O_1$, $o_2 \in I^O_2$,
    and $\running \subseteq (P \setminus \{p_1,p_2\})$ such that $ C \subseteq R$.
    ~~Then, for any $c \in C$:
   \begin{align*}
    &\pi(\running, c) - 
    \pi(\running \cup \{p_1\}, c) \\
    &=\pi(\running \cup \{p_2\}, c) 
    - \pi(\running \cup \{p_2, p_1\}, c)
    \end{align*}

\end{lemma}

\section{Control by Adding Parties}\label{sec:adding}

\subsection{Technical Proofs for SSP}
\Cref{lemma:left-most-right-most-equal} proves that, in each coalition-interval, it is sufficient to consider the addition of at most two spoiler parties.

\begin{lemma}\label{lemma:left-most-right-most-equal}
    Let $I^C \in \mathcal{I}^C$ be a coalition-interval.
    For any subset $K \subseteq S \cap I^C$, there exists $K' \subseteq K$ with $|K'| \leq 2$, such that for $R(K) := (\universe \setminus S) \cup K$ and $R(K'):= (\universe \setminus S) \cup K'$:
    \begin{align*}
        \sum_{p \in C \cap R(K)} \pi(R(K), p) = \sum_{p \in C \cap R(K')} \pi(R(K'), p) 
    \end{align*}
\end{lemma}

\begin{proofsketch}
    If $|K|\leq 2$, then the claim is trivial. 

    Assume that $|K| >2$. 
    Notice that all parties in $K$ are coalition parties.
    Let $c_L$ and $c_R$ be the left-most and right-most parties in $K$, respectively; and consider $K'$ that contains these two parties, that is:
    \begin{align*}
        K' := \{c_L \mid c_L < c, \forall c \in K\} \cup \{c_R \mid c < c_R, \forall c \in K\}
    \end{align*}

    We prove that if a voter's favorite party is in the coalition under $R(K)$, it will also be in the coalition under $R(K')$.
    If its favorite party is in $R(K')$, it will remain their favorite, as it is still available and there are even fewer options to choose from. Otherwise, we prove that the voter prefers either $c_L$ or $c_R$ (or both) over all the opposition parties. This ensures that their vote will still be given to the coalition under $R(K')$.

    Conversely, we show that if a voter's favorite party is in the opposition under $R(K)$, it will remain their favorite party under $R(K')$.
    This follows from the fact that $(P\setminus C) \subseteq R$, thus their favorite party is also in $R(K')$, with even fewer options available.
    
\end{proofsketch}


    


    

\begin{lemma}\label{lemma:adding-opposition-two}
    Let $I^O \in \mathcal{I}^O$ be an opposition-interval.   
    for any subset $K \subseteq S \cap I^O$, there exists $K' \subseteq K$ with $|K'| \leq 2$, such that for $R(K) := (\universe \setminus S) \cup K$ and $R(K'):= (\universe \setminus S) \cup K'$:
    \begin{align*}
        \sum_{p \in C \cap R(K)} \pi(R(K), p) = \sum_{p \in C \cap R(K')} \pi(R(K'), p) 
    \end{align*}
\end{lemma}

The proof is similar to the one of \Cref{lemma:left-most-right-most-equal} and provided in the appendix.

\subsection{Adding Coalition Parties }
For general preferences, \Cref{thm:acp-general-C} proves that the problem is W[2]-hard even when the chair cares only about the coalition objective ($\rho =0$).

For SSP preferences, \Cref{thm:acp-ssp-CFP-non-con} proves that the problem is W[1]-hard under the compact representation. 
However, \Cref{thm:acp-ssp-C} provides a polynomial-time algorithm for $\rho =0$; while \Cref{thm:acp-ssp-CFP-con} provides a polynomial-time algorithm for the case of a contiguous coalition ($q=1$). 
Then, \Cref{thm:CCFP-ACP-pseudo-polynomial} presents a pseudo-polynomial algorithm under the compact representation, which becomes polynomial under the extensive representation.


\subsubsection{General Preferences}


\begin{theorem}\label{thm:acp-general-C}
    \withoutTwithoutLacp\ is W[2]-hard.
\end{theorem}

We provide a reduction from the dominating set problem.
 Clearly, \withoutTwithLacp\ is also w[2]-hard as  $\rho$ can be set to zero, and then it solves exactly the \withoutTwithoutLacp\ problem.


\subsubsection{Symmetric Single-Peaked Preferences}


\begin{theorem} \label{thm:acp-ssp-CFP-non-con}
   For SSP preferences, \withoutTwithLacp\ is NP-hard under the compact representation.
\end{theorem}

The proof is by reduction from the k-subset-sum problem.



\begin{theorem} \label{thm:acp-ssp-C}
    For SSP preferences, \withoutTwithoutLacp\ is polynomial-time solvable.
\end{theorem}
\begin{proof}

    We present a dynamic programming algorithm to determine the maximum number of votes a coalition can achieve by adding at most $k$ coalition parties. Using this result, we can verify whether the coalition objective is attainable.

    The algorithm is based on two key observations. First, by \Cref{cor:interval-adding}, adding parties in one coalition-interval does not affect the impact of adding parties into other coalition-intervals. 
    Second, by \Cref{lemma:left-most-right-most-equal}, it suffices to consider subsets of size at most $2$ within each interval to maximize votes.

    Combining these results, we restrict the set of combinations to be evaluated: at most two parties can be added per interval, and the total number of added parties does not exceed $k$. This restricted set of combinations can be efficiently handled using dynamic programming, as described below.

    \paragraph{Algorithm.} The algorithm consists of three steps. 
    
    \paragraph{(1) Per-Interval Calculation.} 
    For each $q' = 1,\ldots,q$ and for $k' = 0,\ldots, k$, we compute the maximum number of votes that can be obtained by the coalition by adding at most $k'$ parties in the interval $I_{q'}$, denoted by $N'_C(I_{q'},k')$.
    
    For $k' \leq 2$, we directly check all subsets of size at most $k'$.
    For $k'\geq 3$, by \Cref{lemma:left-most-right-most-equal}, the maximum number of votes can be obtained by considering at most two parties, so we simply use the value in $N'_C(I_{q'},2)$.
    Clearly, this step can be done in polynomial time.

    \paragraph{(2) Dynamic Programming.} 
    For each $q' = 1,\ldots,q$ and for $k' = 0,\ldots, k$, we calculate the maximum number of votes that can be obtained by the coalition by adding at most $k'$ parties in the intervals $\mathcal{I}_{q'}:= \{I_1,\ldots, I_{q'}\}$, denoted by $N_C(\mathcal{I}_{q'},k')$.
    \begin{itemize}
        \item Base Cases: For $q'=1$:
        \begin{align*}
            N_C(\mathcal{I}_{1},k') := N'_C(I_{1},k')
        \end{align*}

        \item Recursive Step: For $q' \geq 2$:
        \begin{align*}
            N_C(\mathcal{I}_{q'},k'):= \displaystyle \max_{0\leq i \leq \min(2,k')}(&N'_C(I_{q'}, i) \\
            &+ N_C(\mathcal{I}_{q'-1},k' -i) )
        \end{align*}
    \end{itemize}

    \paragraph{(3) Objective Check.} We check whether the coalition objective is satisfied when the number of votes for the coalition is $N_C(\mathcal{I}_{q},k)$ -- that is, whether:
    \begin{align*}
        N_C(\mathcal{I}_{q},k) / n \geq \varphi
    \end{align*}
    If so, assert "yes"; otherwise assert "no".
%
%
\end{proof}
The proof of correctness is provided in the appendix. 




\begin{theorem} \label{thm:acp-ssp-CFP-con}
    For SSP preferences, \withoutTwithLacp\ is polynomial-time solvable for contiguous coalitions ($q=1$).
\end{theorem}

\begin{proof}
       We consider any subset of $S$ of size at most $min\{2, k\}$. 
    For each subset, we check if it satisfies the objectives. 
    If it does, the algorithm stops and asserts "yes". 
    If no subset satisfies the objectives, the algorithm asserts "no".  

    It is clear that this process takes polynomial time, as the number of subsets is less than $m^2$.
    We shall now see that the algorithm is always correct.
    
    When the algorithm asserts "yes", there indeed exists a subset of size at most $k$ that satisfies the objectives as the algorithm asserts "yes" only when it finds one.  

    When the algorithm returns "no", it means that no subsets of size at most $\min\{2, k\}$ satisfy the objectives. 
    If $k \leq 2$, the algorithm is correct, as it directly checks all possible subsets. 
    Otherwise, $k > 2$. Assume for contradiction that there exists a subset $K$ with $|K| \leq k$ that satisfies the objectives. 
    Since the algorithm asserted "no", it must be that $|K| > 2$.  

    Since the coalition is contiguous, all parties in $K$ are in the same coalition-interval. By \Cref{lemma:left-most-right-most-equal}, there exists a set $K' \subseteq K$ with $|K'| \leq 2$ such that adding the parties in $K'$ instead of $K$ results in the coalition receiving the same fraction of votes.  
    Since \( K \) satisfies the coalition objective, \( K' \) must also does.

    On the other hand, as the set of running parties when adding $K'$ is a subset of the running parties when adding $K$, it follows from \Cref{cor:adding-parties-decreases} that the fraction of any running party when adding $K'$ is at least as high as its fraction when adding $K$.
    This holds in particular for the favorite party. 
    As \( K \) satisfies the favorite-party objective, \( K' \) also does.  

    Thus, $K'$ satisfies both objectives.
    However, since $|K'| = 2$, the algorithm must have considered and checked $K'$. 
    But the algorithm asserted "no", which means that $K'$ did not satisfy the objectives -- a contradiction.  
\end{proof}



\begin{theorem}\label{thm:CCFP-ACP-pseudo-polynomial}
    For SSP preferences, \withoutTwithLacp\ admits a pseudo-polynomial algorithm for the compact representation.
\end{theorem}

Notice that such an algorithm runs in polynomial time under the extensive representation.

\begin{proof}
The algorithm is conceptually similar to the one presented in \Cref{thm:acp-ssp-C} but with a key difference: instead of computing the maximum number of votes attainable by the coalition, it evaluates whether specific vote combinations meet the objectives.

    Formally, we define the set $A$ of acceptable vote combinations, representing pairs of votes for the coalition and the favored party that meet the objectives, as follows:
    \begin{align*}
    A := \left\{ (n_c, n_f) \in \{1, \ldots, n\}^2\ \middle\vert \begin{array}{l}
        \frac{1}{n}(n_f + n_c) \geq \varphi\\
        n_f \geq \rho \cdot n_c
      \end{array}\right\}
    \end{align*}

    The goal is to check whether such vote combination can be obtained by adding at most $k$ coalition parties from $S$. 
    
    As before, the algorithm is based on two key observations. \Cref{cor:interval-adding} ensures that adding parties in one coalition-interval  does not affect the impact  of adding parties into other coalition-intervals; 
    while \Cref{lemma:left-most-right-most-equal}, ensures that it suffices to consider only subsets of size at most $2$ in each interval.
    
    \paragraph{Algorithm.} The algorithm consists of three steps. 
    
    \paragraph{(1) Per-Interval Calculation.} 
    For each $q' = 1,\ldots,q$, ~~$k' = 0,\ldots ,k$,~~ $n'_c = 0, \ldots, n$,  and $n'_f = 0, \ldots, n$; we check whether it is possible to achieve exactly $n_c$ votes for the coalition and $n_f$ votes for the favored party, by adding at most $k'$ parties in the interval $I_{q'}$.
   To do so, we employ \Cref{lemma:left-most-right-most-equal} and \Cref{cor:adding-parties-decreases}, which imply that any number of votes for the coalition, $n_c$, that can be achieved by a certain number of parties can also be achieved by at most two parties while ensuring the maximum possible votes for the favored party, $n_f$, corresponding to this $n_c$.
    Clearly, this step can be done in polynomial time.
    
    The result, a "Yes" or "No" answer, is denoted by $W'(I_{q'},k',n'_c, n'_f)$.

    \paragraph{(2) Dynamic Programming.} 
    
    For each $q' = 1,\ldots,q$,~~ $k' = 0,\ldots, k$,~~ $n'_c = 0, \ldots, n$,  and $n'_f = 0, \ldots, n$; we check whether it is possible to achieve exactly $n_c$ votes for the coalition and $n_f$ votes for the favored party by adding at most $k'$ parties in the intervals $\mathcal{I}_{q'}:= \{I_1,\ldots, I_{q'}\}$, as follows. 
    \begin{itemize}
         \item Base Case:  For $q'=1$: ~~
         \begin{align*}
             W(\mathcal{I}_{1},k',n'_c, n'_f) =  W'(I_{1},k',n'_c, n'_f)
         \end{align*}
         
        \item Recursive Step: For $q' \geq 2$, 
         $W(\mathcal{I}_{q'},k',n'_c, n'_f) = "yes"$ if and only if there exist $k''<k'$, $n''_c <n'_c$ and $n''_f <n'_f$
          such that $W'(I_{q'},k'',n''_c, n''_f) = "yes"$ and $W(\mathcal{I}_{q'-1},k'-k'',n'_c-n''_c, n'_f-n''_f) = "yes"$.
       
    \end{itemize}

    \paragraph{(3) Objective Check.} We check, for $k'=0,\ldots k$ and any acceptable vote combination $(n'_c, n'_f) \in A$, whether the answer in $W(\mathcal{I}_{q},k',n'_c, n'_f)$ is "yes".
    If so, then the algorithm asserts "yes"; otherwise it asserts "no".

    \paragraph{Time Complexity.}
    The time complexity is $\mathcal{O}(p\times k \times n^2)$, which is polynomial in the extensive representation, but only pseudo polynomial in the compact representation (where the input size does not depend on $n$).
\end{proof}

\subsection{Adding Opposition Parties}

\Cref{thm:-ccaop-immune} shows that when the chair cares only about the coalition objective ($\rho =0$), achieving it using this type of control is not possible (i.e., the problem is immune).

In contrast, for $\rho >0$, \Cref{thm:aop-general-CFP} shows that the problem is W[2]-hard for general preferences.
\Cref{thm:aop-ssp-CFP-non-con} proves that the problem is NP-hard even for SSP preferences, but only under the compact representation.
However, \Cref{thm:aop-ssp-CFP-con} provides a polynomial-time algorithm for contiguous coalitions ($q=1$). 
Lastly, \Cref{thm:CCFP-AOP-pseudo-polynomial} presents a pseudo-polynomial algorithm under the compact representation.







\begin{theorem}\label{thm:-ccaop-immune}
     \withoutTwithoutLaop\ is immune.
\end{theorem} 

\begin{proof}
    Each voter gives exactly one point to the party ranked first; hence, adding parties to other positions makes no difference.
    We distinguish between two possible cases:
    \begin{enumerate}
        \item  For a voter whose first-ranked party is not from the coalition, adding a spoiler opposition party to the first position does not change the outcome, as the total points received by opposition parties remain the same.
        \item For a voter whose first-ranked party is from the coalition, adding a spoiler opposition party to the first position reduces the points received by the coalition.  
    \end{enumerate}
    
    This problem concerns only the percentage of points received by the coalition as a whole, rather than by any specific favored party within it.
    As a result, adding spoiler opposition parties can only harm the coalition's overall standing.
\end{proof}


\subsubsection{General Preferences}


\begin{theorem} \label{thm:aop-general-CFP}
     \withoutTwithLaop\ is W[2]-hard.
\end{theorem}

The proof is by reduction from the dominating set problem. 

\subsubsection{Symmetric Single-Peaked Preferences}


\begin{theorem} \label{thm:aop-ssp-CFP-non-con}
   For SSP preferences, \withoutTwithLaop\ is NP-hard under the compact representation.
\end{theorem}

The proof is by reduction from the k-subset sum problem.  




\begin{theorem}\label{thm:aop-ssp-CFP-con}
     For SSP preferences, \withoutTwithLaop\ is polynomial-time solvable for contiguous coalitions ($q=1$).
\end{theorem}

\begin{proofsketch}
    We prove that it is sufficient to consider any subset $S' \subseteq S$ with $|S'| \leq k$, $|S' \cap I^O_0 | \leq 1$ and $|S' \cap I^O_1 | \leq 1$ -- that is, at most one party is added from the left of the coalition and at most one from the right. 
    For each subset, we check if it satisfies the objectives. 
    If it does, the algorithm stops and asserts "yes". 
    If no subset satisfies the objectives, the algorithm asserts "no". 
\end{proofsketch}

\begin{theorem}\label{thm:CCFP-AOP-pseudo-polynomial}
    For SSP preferences, \withoutTwithLaop\ admits a pseudo-polynomial algorithm for the compact representation.
\end{theorem}

\begin{proofsketch}
    The algorithm is similar to the one presented in \Cref{thm:CCFP-ACP-pseudo-polynomial}; only that the dynamic programming check the possible combination of adding at most $2$ parties in each opposition-interval while the total size is at most $k$.    
    The correctness follows from \Cref{lem:sperate-opposition-intervals}, which ensures that different opposition-intervals can be treated separately; and \Cref{lemma:adding-opposition-two}, which ensures that it suffices to consider only subsets of size at most $2$ in each opposition-interval.    
\end{proofsketch}

\section{Control by Deleting Parties}\label{sec:deleting}
\subsection{Technical Proofs for SSP}
By Lemma~\ref{lemma:move-to-nerest_delete}, deleting a party can only transfer its votes to one of its adjacent parties; hence in order to move votes from the coalition to the opposition and vice versa, the deleted parties must be adjacent to one end of the interval. Formally:
\begin{lemma}\label{lemma:most-adjacent-only}
    Let $I$ be an opposition (coalition) interval.
    For any $K\subseteq P\cap I$, there exists $K'\subseteq K$ such that: (1) The parties in $K'$ are adjacent to one end of the interval. Specifically, there is no party $k_i \in K'$ with the two parties on its left and right sides are opposition (coalition) parties not in $K'$.
    (2) Let $R(K'):= P \setminus K'$ and $R(K):= \universe \setminus K$:
    \begin{align*}
        \sum_{p \in C} \pi(R(K), p) = \sum_{p \in C } \pi(R(K'), p) 
    \end{align*}
\end{lemma}
\begin{proof}
    By Lemma~\ref{lemma:move-to-nerest_delete}, deleting an opposition (coalition) party surrounded by other opposition (coalition) parties on both sides will transfer the votes to an opposition (coalition) parties, resulting in the same change to the sizes of the coalition and the opposition, as the group without that party.
\end{proof}

\subsection{Deleting Coalition Parties}

\Cref{thm:-ccdcp-immune} proves that when $\rho =0$, thr problem is immune. 
In contract, when a favored party is also taken under consideration ($\rho >0$), the objectives can be achieved.

For $\rho >0$, the problem is W[1]-hard due to \Cref{thm:dcp-general-CFP}. 
\Cref{thm:dcp-ssp-CFP-non-con} proves that the problem is NP-hard even for SSP preferences, but only under the compact representation. 
\Cref{thm:dcp-ssp-CFP-con} shows that the problem is polynomial-time solvable for contiguous coalitions ($q=1$). 
\Cref{thm:CCFP-dcp-pseudo-polynomial} presents a pseudo-polynomial algorithm for the compact representation, which is a polynomial-time for the extensive one.



\begin{theorem}\label{thm:-ccdcp-immune}
     \withoutTwithoutLdcp\ is immune.
\end{theorem} 
\begin{proofsketch}
Each voter awards one point to their top-ranked party. Deleting coalition parties does not affect the points received by opposition parties for voters whose first-ranked party is outside the coalition. For voters whose first-ranked party is in the coalition, deleting that party may decrease the coalition's points if their next preference is an opposition party. Since the problem concerns the coalition's total percentage of points rather than any specific favored party, deleting opposition parties can only reduce the coalition's overall standing. 
\end{proofsketch}
More details of the proof are provided in the appendix. 

\subsubsection{General Preferences}


\begin{theorem} \label{thm:dcp-general-CFP}
    \withoutTwithLdcp\ is W[1]-hard.
\end{theorem}

The proof is by reduction from the k-clique problem.

\subsubsection{Symmetric Single-Peaked Preferences}



\begin{theorem} \label{thm:dcp-ssp-CFP-non-con}
   For SSP preferences, \withoutTwithLdcp\ is NP-hard  under the compact representation.
\end{theorem}
The proof is by reduction from the $k$-subset sum problem.




\begin{theorem} \label{thm:dcp-ssp-CFP-con}
    For SSP preferences, \withoutTwithLdcp\ is polynomial-time solvable for contiguous coalitions ($q=1$).
\end{theorem}

\begin{proofsketch}
    Deleting coalition parties can either shift votes to the favored party or reduce the coalition's size, which might help achieve the favored-party objective. Thus, we consider two specific types of deletions:
    (1) deleting parties near the favored party, and (2) deleting parties near the two endpoints of the coalition-interval; causing some or all of their votes to shift to the opposition.
    By Lemma~\ref{lemma:most-adjacent-only}, it suffices to consider continuous deletions within each of these types. 

    Given that $q=1$, the number of locations to consider for deletions is limited to $4$.
    We can delete up to $k$ parties. For each $0\leq k' \leq k$, we evaluate all possible ways to partition the $k'$  deletions among these $4$ locations. If any such partition satisfies both objectives, the algorithm returns "yes" Otherwise, it returns "no".

    Since the number of possible partitions of the deleted parties that need to be evaluated is polynomial, the algorithm runs in polynomial time.
\end{proofsketch}

    



\begin{theorem}\label{thm:CCFP-dcp-pseudo-polynomial}
    For SSP preferences, \withoutTwithLdcp\ admits a pseudo-polynomial algorithm for the compact representation.
\end{theorem}

\begin{proofsketch}
    Building on the polynomial algorithm for \Cref{thm:dcp-ssp-CFP-con} and the key observations from \Cref{cor:interval-adding}—which state that deleting parties from one coalition-interval does not affect the impact of deletions from other coalition-intervals, we can apply dynamic programming techniques similar to those used in the proof of \Cref{thm:CCFP-ACP-pseudo-polynomial}.
\end{proofsketch}

\subsection{Deleting Opposition Parties}

\Cref{thm:dop-general-C} proves that the problem is W[1]-hard even when $\rho =0$.
For SSP preferences, \Cref{thm:dop-ssp-CFP-non-con} proves that the problem is W[1]-hard under the compact representation. 
\Cref{thm:dop-ssp-C} provides a polynomial-time algorithm for $\rho =0$; while \Cref{thm:dop-ssp-CFP-con} provides such algorithm for contiguous coalitions ($q=1$). 
\Cref{thm:CCFP-DOP-pseudo-polynomial} presents a pseudo-polynomial algorithm under the compact representation.

\subsubsection{General Preferences}


\begin{theorem} \label{thm:dop-general-C}
    \withoutTwithoutLdop\ is W[1]-hard.
\end{theorem}

The proof is by reduction from the $k$-clique problem. 
Clearly, \withoutTwithLdop\ is also w[1]-hard as  $\rho$ can be set to zero, and then it solves exactly the \withoutTwithoutLdop\ problem.

\subsubsection{Symmetric Single-Peaked Preferences}



\begin{theorem} \label{thm:dop-ssp-CFP-non-con}
   For SSP preferences, \withoutTwithLdop\ is NP-hard under the compact representation.
\end{theorem}
The proof is by reduction from the k-subset sum problem.



\begin{theorem} \label{thm:dop-ssp-C}
    For SSP preferences, \withoutTwithoutLdop\ is polynomial-time solvable.
\end{theorem}
\begin{proofsketch}
    We present a dynamic programming algorithm to determine the maximum votes a coalition can achieve by deleting up to $k$ opposition parties, allowing us to verify if the coalition's objective is attainable.

    The algorithm relies on two observations: (1) By \Cref{lem:sperate-opposition-intervals}, deleting parties from one opposition-interval does not affect the impact of deleting an opposition party from another opposition-interval. (2) By \Cref{lemma:most-adjacent-only}, to maximize votes that move to the coalition by deleting $k'$   parties from one opposition-interval, it suffices to consider their distribution between the two sides of that opposition-interval.

    This restricts the combinations to evaluate, allowing us to use dynamic programming techniques.
\end{proofsketch}

\begin{theorem} \label{thm:dop-ssp-CFP-con}
    For SSP preferences, \withoutTwithLdop\ is polynomial-time solvable for contiguous coalitions ($q=1$).
\end{theorem}

\begin{proofsketch}
The goal is to check whether deleting at most $k$ opposition parties can achieve a valid combination of additional votes for the coalition and the favored party.

By \Cref{lem:sperate-opposition-intervals}, deletions in one opposition-interval do not affect the impact of deletions in the other.
By \Cref{lemma:most-adjacent-only}, we only need to consider deletions from the two sides of each interval.The approach involves distributing the $k$ deletions between the two intervals $(k_0, k_1)$. For each interval, we compute all possible vote combinations resulting from deletions by considering every way to split the $k_i$ deletions between the two endpoints of the interval.
Then, check whether any pair of vote combinations from the two intervals satisfies both objectives while ensuring $k_0+k_1 \leq k$.
\end{proofsketch}

\begin{theorem}\label{thm:CCFP-DOP-pseudo-polynomial}
    For SSP preferences, \withoutTwithLdop\ admits a pseudo-polynomial algorithm for the compact representation.
\end{theorem}

\begin{proofsketch}
    Building on the polynomial algorithm for \Cref{thm:dop-ssp-CFP-con} and the key observations from \Cref{lem:sperate-opposition-intervals}—which state that deleting parties from one opposition-interval does not affect the impact of deletions from other opposition-intervals, we can apply dynamic programming techniques similar to those used in the proof of \Cref{thm:CCFP-ACP-pseudo-polynomial}. 
\end{proofsketch}

\section{Conclusion and Future Work}\label{sec:conclusion}
This paper introduces a novel problem of control in parliamentary elections, revealing both the inherent complexity of this problem and the unique structure of symmetric single-peaked preferences, which make many types of such control possible and computable. This work lays the foundation for a wide range of open questions.

\paragraph{Alternative Election Rules.} 
While this paper focuses on plurality voting, an important future research is exploring the complexity under alternative rules, such as Borda, Condorcet, and Approval.

\paragraph{New Types of Control.}
Expanding the types of control for parliamentary election is particularly intriguing. Specifically, studying control through splitting and unifying parties - which are common in practice. 
Another direction is destructive control, where the chair aims to prevent a rival-coalition from succeeding. 
It would also be interesting to study scenarios where the chair not only has a favored party but specifies a desired balance of power among the coalition parties.


\paragraph{Beyond Control.} Extending the framework to study other ways to influence parliamentary elections — such as manipulation and bribery — is another natural and compelling next step.





\section{Acknowledgments}
This research is partly supported by the Israel Science Foundation grants 2544/24, 3007/24 and 2697/22.

\clearpage
\bibliographystyle{ijcai-format/named}
\bibliography{main}

\clearpage
\appendix

\section{Organization of the Appendix}
\Cref{appendix:sec:notations} summarizes the notations used throughout the paper.
\Cref{sec:examples} presents examples of the two types of objectives. 
\Cref{app:types} focuses on symmetric single-peaked preferences.
\Cref{appendix:sec:lammas} contains the proofs of the technical lemmas.
\Cref{appendix:sec:poly} provides the details of the polynomial algorithm, while \Cref{apx:immune} focuses on the immune proofs, and \Cref{appendix:sec:hard} includes the hardness proofs.

\section{Table of Notations} \label{appendix:sec:notations}
See Table~\ref{tab:notation} for a summary of notations, and \Cref{tab:problem_shorthand} for the full names of the problems.
\begin{table}[H]
    \centering
  
    \begin{tabular}{|p{2cm}|p{6cm}|}
        \hline
        \textbf{Notation} & \textbf{Description} \\ \hline
        $\universe$ & Set of all parties (or candidates) that could potentially run in the election. \\ \hline
        $V$ & Set of $n$ voters, $V= \{v_1, \ldots, v_n\}$. \\ \hline
        $\succ_v$ & Strict preference order of voter $v$ over the parties in $\universe$. \\ \hline
        $\running$ & Set of parties running in the election, $\running \subseteq \universe$. \\ \hline
        $\topvR$ & Top (most-preferred) party of voter $v$ in set $R$. \\ \hline
        $N(\running, p)$ & Number of voters whose top choice in $\running$ is party $p$. \\ \hline
        $\pi(\running, p)$ & Fraction of votes received by party $p$ out of the total number of votes, \\ 
        & $\pi(\running, p) := N(\running, p)/n$. \\ \hline
        $C$ & Target-coalition, a subset of $\universe$. \\ \hline
        $\varphi$ & Coalition target-fraction. \\ \hline
        $p_F$ & Favored-party within the coalition $C$. \\ \hline
        $\rho$ & Target-favored-party-ratio. \\ \hline
        $k$ & Integer representing the maximum number of parties that can be deleted or added. \\ \hline
        $S$ & Set of spoiler parties that can be added by the chair. \\ \hline
        $\ell$ & Number of spoiler parties in set $S$. \\ \hline
        $\mathcal{I}^C := \{I^C_1, \ldots, I^C_q\}$ & The set of coalition intervals. If $q = 1$, the coalition is contiguous.  \\ \hline
        $\mathcal{I}^O := \{I^O_0, \ldots, I^O_{q+1}\}$ & The set of opposition intervals. \\ \hline

    \end{tabular}
      \caption{Notation and Description}
    \label{tab:notation}
\end{table}
\begin{table}[htbp]
    \centering
     \begin{tabular}{|p{2cm}|p{6cm}|}
        \hline
        \textbf{Shorthand} & \textbf{Full Name} \\ \hline
        CC-ACP             & Control for Coalition by Adding Coalition Parties \\ \hline
        CC-AOP             & Control for Coalition by Adding Opposition Parties \\ \hline
        CC-DCP             & Control for Coalition by Deleting Coalition Parties \\ \hline
        CC-DOP             & Control for Coalition by Deleting Opposition Parties \\ \hline
        CCFP-ACP           & Control for Coalition and Favored Party by Adding Coalition Parties \\ \hline
        CCFP-AOP           & Control for Coalition and Favored Party by Adding Opposition Parties \\ \hline
        CCFP-DCP           & Control for Coalition and Favored Party by Deleting Coalition Parties \\ \hline
        CCFP-DOP           & Control for Coalition and Favored Party by Deleting Opposition Parties \\ \hline
    \end{tabular}
    \caption{Shorthand and Full Names of Problems}
    \label{tab:problem_shorthand}
\end{table}

\section{Simple Examples} \label{sec:examples}
\begin{Example} [\withoutTwithoutLacp\ problem]
    Let the set of parties be $P = \{p_1, p_2, p_3\}$ and the coalition be $C = \{p_1, p_2\}$.
    Let the set of voters be $V = \{v_1,v_2,\ldots ,v_6\}$. 
    The preference orders of $v_1$ and $v_2$ are $p_3 \succ p_2 \succ p_1$.  
    and the preference orders of $v_3$ and $v_4$ are $p_2\succ p_3 \succ p_1$. 
    and the preference orders of $v_5$ and $v_6$ are $p_1\succ p_3 \succ p_2$.    
    The target support $\varphi = \frac{2}{3}$.
    
    Let $k = 1$ and the set of spoiler party be $S=\{p_2\}$.
    
    By adding the spoiler coalition party $p_2$, the voters $v_3$ and $v_4$ vote for a coalition party ($p_2$) instead of voting for an opposition party ($p_3$).
    Thus, the coalition receives $\frac{2}{3}$ of the votes as required. 
\end{Example}

\begin{Example}[\withoutTwithLacp\ problem]
    Let the set of parties be $P = \{p_1, p_2, p_3, p_4\}$ the coalition be $C = \{p_1, p_2,p_3\}$, and the favored party be $p_1$.
    Let the set of voters be $V = \{v_1,v_2,\ldots ,v_6\}$. 
    The preference orders of $v_1$ and $v_2$ are $p_3 \succ p_4\succ p_2 \succ p_1$.  
    and the preference orders of $v_3$ and $v_4$ are $p_2\succ p_3 \succ p_4\succ p_1$. 
    and the preference orders of $v_5$ and $v_6$ are $p_1\succ p_3 \succ p_2 \succ p_4$.    
    The target support $\varphi = \frac{2}{3}$, and the target ratio $\rho = \frac{1}{2}$.
    Let $k = 1$ and the set of spoiler party be $S=\{p_2, p_3\}$.

    By adding the spoiler coalition party $p_2$, voters $v_3$ and $v_4$ vote for a coalition party ($p_2$) instead of voting for an opposition party ($p_4$).
    The coalition receives $4$ votes, and the favored party, $p_1$, receives $2$ votes from these votes. Thus, both requirements are met. 

    Note that adding the spoiler coalition party $p_3$, is not a successful control action.
    By adding $p_3$, the $4$ voters $v_1 \ldots, v_4$ vote for a coalition party ($p_3$) instead of voting for an opposition party ($p_4$).
    The coalition receives $6$ votes and the favored party, $p_1$, receives only $2$ votes from these votes.
    Thus, only the requirement of the target support $\varphi$ is met.
\end{Example}

\section{Types for Voters in Symmetric Single-Peaked Preferences}
\label{app:types}
In this section we prove that for symmetric single-peaked preferences, the set of voters can be fully represented in a more efficient way --- by $\mathcal{O}(m^2)$ values, where $m$ is the number of parties in the election; this allows us to be able to work with a number of voters which is exponential to the problem size.
\paragraph{Efficient Representation.} We start by proving that there are only $(n^2+1)$ different symmetric single-peaked preference orders.

\newcommand{\divP}[2]{d(#1,#2)}

\begin{definition}[Divider]
    Let $p_1, p_2 \in P$. Their average, denoted by $\divP{p_1}{p_2} := \frac{1}{2}(p_1 + p_2) $, is referred to as their divider.
\end{definition}

The importance of the divider of two parties is as follows. 
\begin{lemma}\label{duplicate}
    Let $v \in V$ and $p_1 < p_2 \in P$. Then, $p_1 \succ_v p_2$ if-and-only-if $ v < \divP{p_1}{p_2}$.
\end{lemma}

\begin{proof}
    Assume that $p_1 \succ_v p_2$.
    Then, by definition, $|v - p_1| < |v - p_2|$.
    Consider the three following cases.
    
    First, if $v < p_1$ then it is clear that $ v < \frac{1}{2} (p_1+p_2)$.

    If $p_1 < v < p_2$.
    As $p_1 < v$, we get that  $|v - p_1| =  v - p_1$.
    However, $v < p_2$, which means that $|v - p_2| = p_2 - v$.
    Thus, $ v < \frac{1}{2} (p_1+p_2)$:
    \begin{align*}
        &v - p_1 = |v - p_1| < |v - p_2| = p_2 - v\\
        &\Rightarrow 2v < p_1 + p_2
    \end{align*}

    Otherwise, $p_2 < v$.
    As before, $p_1 < v$, which means that $|v - p_1| =  v - p_1$.
    But now also, $p_2 < v $; thus $|v - p_2| = v - p_2$.
    Since $p_1 < p_2$, we can conclude that $|v - p_1|> |v-p_2|$ -- a contradiction: 
    \begin{align*}
        & |v-p_1| = v - p_1 >  v - p_2 = |v - p_2| 
    \end{align*}

    On the other hand, assume that $ v < \frac{1}{2} (p_1+p_2)$
\end{proof}

We start by defining the set of \emph{dividers}: $$D := \{ \frac{p_1 + p_2}{2} \mid p_1, p_2 \in \universe\} \cup \{0,1\}$$
and then denote the dividers in $D$ in ascending order by $d_1, \ldots, d_{|D|}$. 
Next, we define the set of \emph{bands}, where each band lies between a pair of consecutive dividers:
\begin{align*}
    B := \{[d_i, d_{i+1}] \mid i = 1,\ldots, |D|-1\}
\end{align*}
We prove that the set of voters can be fully represented by specifying, for each band, what is the number of voters whose peak lies within it.
Accordingly, in this setting, we assume that the preferences are provided as 
a list of bands along with the corresponding number of voters in each band.
Notice that the number of bands depends only on the number of parties.


%
%
%

\begin{lemma}\label{lemma:bands}
    Let $v_1,v_2 \in V$ be two voters on the same band $[d_i, d_{i+1}] \in B$. Then, the preference orders of $v_1$ and $v_2$ are identical.
\end{lemma}

\begin{proof}
    Let $p_1, p_2 \in \universe$ be two parties such that $p_1 < p_2$.
    
    We prove that it is either the case where both $v_1$ and $v_2$ prefer $p_1$ over $p_2$ or that both prefer $p_2$ over $p_1$.

    Let $d(1,2) := \frac{p_1 + p_2}{2}$ be the divider induces by parties $p_1$ and $p_2$. 
    Notice that since we assume that the voters have \emph{strict} preferences, their peaks cannot be on such dividers (as it will imply that the distance to $p_1$ and $p_2$ is equal).

    This divider partitions the interval $[0, 1]$ into two segments: $[0, d(1,2))$ and $(d(1,2), 1]$; where all the points in the first are closer to $p_1$ and all the points in the second are closer to $p_2$. 
    
    Since $v_1$ and $v_2$ are on the same band and since bands are defined as the smallest divisions of $[0,1]$ (as they formed by consecutive dividers), we can conclude that both voter $v_1, v_2$ lie together within $[0, d(1,2))$ or within $(d(1,2), 1]$.
    But this means that either both voters are closer to $p_1$ than to $p_2$ or the opposite. 
    Thus, both $v_1$ and $v_2$ prefer $p_1$ over $p_2$, or that both prefer $p_2$ over $p_1$ -- as required.
\end{proof}

\begin{figure}[h]
    \centering
    \includegraphics{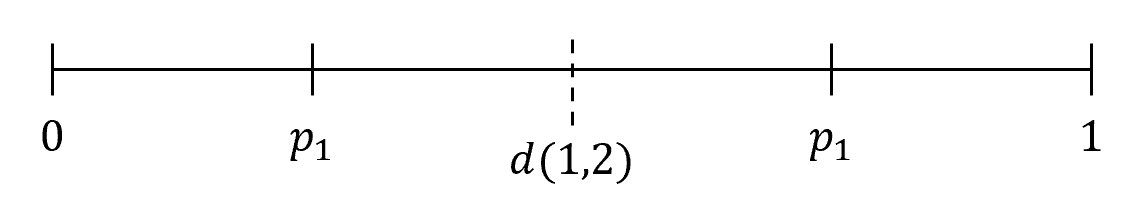}
    \caption{Dividers}
    \label{fig:dividers}
\end{figure}

\section{Proofs for the Technical Lemmas} \label{appendix:sec:lammas}
\begin{reminderlemma}{\ref{lemma:adding-only-decreases}}
    Let $p_1, p_2 \in P$ and $R \subseteq (P\setminus\{p_2\})$ such that $p_1 \in R$. Then,
    \begin{align*}
        \pi(R, p_1) \geq \pi(R \cup \{p_2\}, p_1) 
    \end{align*}
\end{reminderlemma}

 \begin{proof}   
    Since $R$ is a subset of $R \cup \{p_2\}$, adding $p_2$ to the election introduces a new option for voters.
    For any voter whose most-preferred party in $R$ is $p_1$, either (1) $p_1$ is still their most-preferred party in $R \cup \{p_2\}$, or (2) the new party $p_2$ becomes their most-preferred party in $R \cup \{p_2\}$.  
    This implies that the number of voters whose most-preferred party is $p_1$ in $R$ is at least as large as in $R \cup \{p_2\}$.
 \end{proof}

\subsection{SSP Preferences}

\begin{reminderlemma}{\ref{lemma:divider-meaning}}
    Let $v \in V$ and $p_1 < p_2 \in P$. Then, $p_1 \succ_v p_2$ if-and-only-if $ v < \frac{1}{2}(p_1+p_2)$.
    
\end{reminderlemma}
The proof is provided in \Cref{duplicate} in the section of types for voters in symmetric single-peaked preferences in the appendix.


\begin{reminderlemma}{\ref{lemma:add_change_to_added_or_not_change}}
     Let $v \in V$, $p_1, p_2 \in P$ and $R \subseteq (P\setminus\{p_2\})$ such that $p_1 \in R$. If $T_v(R) = p_1$, then $T_v(R\cup \{p_2\}) \in \{p_1, p_2\}$.
\end{reminderlemma}
\begin{proof}
    $T_v(P') = p_i$, hence, from \Cref{lemma:divider-meaning}, $\forall p'\in P'$ it holds that if $p'< p_i$ then $v < \frac{1}{2}(p'+p_i)$, otherwise. $v > \frac{1}{2}(p'+p_i)$  
    The position of $v$ is unchanged, hence $\forall p'\in P'$, it still hold from \Cref{lemma:divider-meaning} that $T_v(P'\cup \{p_j\}) \not= p'$.
\end{proof}
\begin{reminderlemma}{\ref{lemma:move-to-nerest_delete}}
     Let $p_1<\ldots< p_m$ be the common     ordering of the parties on $[0,1]$.
     Let $v \in V$ and $T_v(P) = p_i$. ~~Then, $T_v(P\setminus \{p_i\}) \in \{p_{i-1}, p_{i+1}\}$.
\end{reminderlemma}
\begin{proof}
   
    Since $T_v(P) = p_i$, then from \Cref{lemma:divider-meaning}, 
    $\frac{1}{2}(p_{i-1}+p_i) < v < \frac{1}{2}(p_i+p_{i+1})$.
    Since $\forall j < i-1: p_j < p_{i-1}$ we get that $ v > \frac{1}{2}(p_j+p_{i-1})$.
    Since $\forall j > i+1: p_{i+1} < p_j$ we get that $ v < \frac{1}{2}(p_{i+1}+p_j)$. 

    Hence, $T_v(P\setminus \{p_i\}) \in \{p_{i-1}, p_{i+1}\}$
\end{proof}


\subsubsection{The Coalition}
Recall that a coalition is a subset of parties $C \subseteq P$. 
For SSP preferences, any coalition can be described as a union of \emph{coalition-intervals}. A  coalition-interval is a sub-interval on $[0,1]$ such that each party $p$ which lies in $[a,b]$ is in the coalition $C$.
We assume that the coalition is given as a set of $q$ disjoint coalition-intervals, $\mathcal{I} := \{I_1, \ldots, I_q\}$, with minimum size.  
If $q > 1$, this implies that there must be at least one opposition party that lies between any two coalition-intervals in $\mathcal{I}$.  
When $q=1$, we say that the coalition is \emph{contiguous}.  

For convenience, we will refer to the sub-interval between the coalition-intervals as opposition-interval. Denote the set of these intervals by $\mathcal{I}^O := \{I^O_0, I^O_1, \ldots, I^O_q, I^O_{q+1}\}$.


\subsubsection{Independent Analysis of Coalition-Intervals}

\begin{reminderlemma}{\ref{lem:sperate-intervals}}
Let $I_1\neq I_2 \in \mathcal{I}$ be two coalition-intervals, $p_1 \in I_1$, $p_2 \in I_2$ and $\running \subseteq P\setminus \{p_2\}$ such that $(P\setminus C) \subseteq R$. ~Then:
\begin{align*}
    \pi(\running, p_1) = \pi(\running \cup \{p_2\}, p_1)
\end{align*}

\end{reminderlemma}

\eden{I think this proof should be inside the paper}
\begin{proof}
    First, by \Cref{lemma:adding-only-decreases}, we can conclude that $\pi(\running, p_1) \geq \pi(\running \cup \{p_2\}, p_1)$.
    
    Suppose by contradiction that $\pi(\running, p_1) > \pi(\running \cup \{p_2\}, p_1)$. 
    This means that there exists a voter $v$ who prefers $p_1$ over all parties in $R$, but prefers $p_2$ over $p_1$.
    Since $I_1$ and $I_2$ are two distinct coalition-intervals in $I$, it means that there exists an opposition party, $o \in P\setminus C$, that lies between them.
    Notice that $o \in R$, which means that $v$ prefers $p_1$ over $o$.

    Consider the following two cases.
    If $p_1 < p_2$, it means that $p_1 < o < p_2$. By \Cref{lemma:divider-meaning}, we get that $p_1 \succ_v o$ implies that $v \leq \frac{1}{2}(p_1 + o)$. Since $o< p_2$, this means that $v \leq \frac{1}{2}(p_1 + p_2)$. But this implies that $p_1 \succ_v p_2$ --- a contradiction.     

    Otherwise, $p_2 < p_1$, which means that $p_2 < o < p_1$. Here, $p_1 \succ_v o$ implies that $v > \frac{1}{2}(p_1 + o)$. As $p_2 < o$, this means that $v > \frac{1}{2}(p_1 + p_2)$. But this implies that $p_1 \succ_v p_2$ --- a contradiction. 
\end{proof}

\subsubsection{Independent Analysis of opposition-Intervals}


\begin{reminderlemma}{\ref{lem:sperate-opposition-intervals}}
    Let $I^O_1\neq I^O_2$ be two opposition-intervals, $p_1 \in I^O_1$, $p_2 \in I^O_2$. Let $R\subset P$
    Then:
   \begin{align*}
    \pi(\running, C) - 
    \pi(\running \cup \{p_1\}, C) = \\
    \pi(\running \cup \{p_2\}, C) 
    - \pi(\running \cup \{p_2, p_1\}, C)
    \end{align*}

\end{reminderlemma}

\eden{I'm not sure the description is accurate. First, $R$ is not defined. second, $\pi$ is a function that gets a set of parties and a \textbf{party} (instead of $C$). We can either define it or add notation }
\begin{proof}
Note that deleting an opposition party can only increase the size of the coalition. Hence, the results are non-negative for both sides of the inequalities.
    Suppose by contradiction that
    \begin{align*}
    \pi(\running, C) - 
    \pi(\running \cup \{p_1\}, C) > \\
    \pi(\running \cup \{p_2\}, C) 
    - \pi(\running \cup \{p_2, p_1\}, C)
    \end{align*}
    
    This implies that there exists a voter $v$ who prefers $p_1$, and when $p_1$ does not run, $v$ votes for a coalition party. That is, $T_v(R \cup {p_1}) = p_1$, but $T_v(R) \in C$. However, if $p_2$ runs, this does not occur, and there are two possibilities:
    (1) $T_v(R \cup {p_2, p_1}) \neq p_1$.
    This cannot be true, because $T_v(R \cup {p_1}) = p_1$. If $T_v(R \cup {p_2, p_1}) \neq p_1$, it must be that $T_v(R \cup {p_2, p_1}) = p_2$, since $p_2$ is the only added party. However, since $p_1 \in I^O_1$ and $p_2 \in I^O_2$, there must be at least one coalition party between them. Therefore, it is impossible for adding $p_2$ to cause a shift in the vote from $p_1$ to $p_2$.

    (2) $T_v(R\cup\{p_2\}) \not\in C$. This cannot be since $T_v(R) \in C$, and adding only $p_2$ can only move the vote to $p_2$.
    However, $T_v(R\cup\{p_2\})$ cannot be $p_2$, since we know that  $T_v(R \cup {p_1}) = p_1$, but $T_v(R) \in C$, which imply that $v$ is somewhere in  the opposition-interval $I^O_1$.

    On the other hand, suppose by contradiction that
     \begin{align*}
    \pi(\running, C) - 
    \pi(\running \cup \{p_1\}, C) < \\
    \pi(\running \cup \{p_2\}, C) 
    - \pi(\running \cup \{p_2, p_1\}, C)
    \end{align*}

     It means that when $p_2$ also runs, there exists a voter $v$ who prefers $p_1$ and when $p_1$ does not run $v$ votes to a coalition party. That is, $T_v(\running \cup \{p_2, p_1\} = p_1$, but $T_v(\running \cup \{p_2\}, C)\in C$.
    But if $p_2$ does not run then it does not happen, there are two options:
    (1) $T_v(R\cup\{p_1\}) \not= p_1$. This can not be since $T_v(\running \cup \{p_2, p_1\} = p_1$, and deleting $p_2$ which is not the top of $v$ can not influence on his top.
    (2) $T_v(R) \not\in C$. This can not be from same argument since $T_v(\running \cup \{p_2\}, C)\in C$.

    In both cases, we reach a contradiction, which concludes the proof.

\end{proof}

\section{Missing Polynomial Algorithms}\label{appendix:sec:poly}

\subsection{Adding Coalition Parties}

\begin{reminderlemma}{\ref{lemma:left-most-right-most-equal}}
        Let $I \in \mathcal{I}^C$ be a coalition-interval.
    For any subset $K \subseteq S \cap I$, there exists $K' \subseteq K$ with $|K'| \leq 2$, such that for $R(K) := (\universe \setminus S) \cup K$ and $R(K'):= (\universe \setminus S) \cup K'$:
    \begin{align*}
        \sum_{p \in C \cap R(K)} \pi(R(K), p) = \sum_{p \in C \cap R(K')} \pi(R(K'), p) 
    \end{align*}
\end{reminderlemma}

\begin{proof}
    If $|K|\leq 2$, then the claim is trivial. 

    Assume that $|K| >2$. 
    Notice that all parties in $K$ are coalition parties.
    Let $c_L$ and $c_R$ be the left-most and right-most parties in $K$, respectively; and consider $K'$ that contains these two parties, that is:
    \begin{align*}
        K' := \{c_L \mid c_L < c, \forall c \in K\} \cup \{c_R \mid c < c_R, \forall c \in K\}
    \end{align*}

    We first prove that if a voter's favorite party is in the coalition under $R(K)$, it will also be in the coalition under $R(K')$.
    Thus, the coalition's vote count with $R(K')$ is at least as high as with $R(K)$.
    
    Formally, let $v \in V$ be a voter whose top party in $R(K)$, $\topvRP{v}{R(K)}$,\footnote{Recall that $\topvR$ denotes the top (most preferred) party in $R$.} is in the coalition $C$.
    We prove that $\topvRP{v}{R(K')}$ is in the coalition as well. 
    
    If $\topvRP{v}{R(K)} \in R(K')$, then clearly $\topvRP{v}{R(K')} = \topvRP{v}{R(K)}$ as there are even less options for the voters.
    Otherwise, $\topvRP{v}{R(K)} \notin R(K)$, which means that $\topvRP{v}{R(K)} \in I$. 
    We show that, in this case, voter $v$ prefers either $c_L$ or $c_R$ (or both) over all the opposition parties.
    Therefore, its top party in $R(K')$ must be in the coalition.

    Let $o\in \universe \setminus C $ be an opposition party.
    It follows that $o$ does not lie within the interval $I$, as all parties in $I$ are coalition parties. 
    Let $I = [a,b]$. 
    There are only two cases, either $o < a$ or $b < o$.
    If $o<a$, then $o<a\leq \topvRP{v}{R(K)}$. Since $\topvRP{v}{R(K)} \succ_v o$, by \Cref{lemma:divider-meaning}, we can conclude that $\frac{1}{2} (o+\topvRP{v}{R(K)}) < v$. 
    However, since $c_L \leq \topvRP{v}{R(K)}$, then $\frac{1}{2} (o+c_L) < \frac{1}{2} (o+\topvRP{v}{R(K)}) < v$. As $o < c_L$, this means that $c_L \succ_v o$. 
    
    Otherwise, $b < o$, then $\topvRP{v}{R(K)}\leq b< o$.
    As $\topvRP{v}{R(K)} \succ_v o$, \Cref{lemma:divider-meaning} implies that $v < \frac{1}{2} (\topvRP{v}{R(K)} + o)$. 
    However, since $\topvRP{v}{R(K)} \leq c_R$, then $v < \frac{1}{2} (\topvRP{v}{R(K)} + o) < \frac{1}{2} (c_R + o) $. As $c_R < o$, this means that $c_R \succ_v o$.

    Conversely, if a voter's favorite party is in the opposition under $R(K)$, it will remain their favorite party under $R(K')$.
    This follows from the fact that $(P\setminus C) \subseteq R$, thus its favorite party is also in $R(K')$ and there are even less options available.
    Thus, the coalition's vote count with $R(K')$ is at most as with $R(K)$.  
    
    Combining these two claims, we establish the desired equality.  
\end{proof}

\subsection{Adding Opposition Parties}

\begin{reminderlemma}{\ref{lemma:adding-opposition-two}}
    Let $I^O \in \mathcal{I}^O$ be an opposition-interval.   
    for any subset $K \subseteq S \cap I^O$, there exists $K' \subseteq K$ with $|K'| \leq 2$, such that for $R(K) := (\universe \setminus S) \cup K$ and $R(K'):= (\universe \setminus S) \cup K'$:
    \begin{align*}
        \sum_{p \in C \cap R(K)} \pi(R(K), p) = \sum_{p \in C \cap R(K')} \pi(R(K'), p) 
    \end{align*}
\end{reminderlemma}

\begin{proof}
    If $|K|\leq 2$, then the claim is trivial. 

    Assume that $|K| >2$. 
    Notice that all parties in $K$ are opposition parties.
    Let $o_L$ and $o_R$ be the left-most and right-most parties in $K$, respectively; and consider $K'$ that contains these two parties, that is:
    \begin{align*}
        K' := \{o_L \mid o_L < o, \forall o \in K\} \cup \{o_R \mid o < c_R, \forall o \in K\}
    \end{align*}

    Let $v \in V$ be a voter. 
    First, it is clear that if $v$'s favorite party is in the coalition under $R(K)$, it will remain its favorite party under $R(K')$.
    This follows from the fact that $C \subseteq R$, thus its favorite party is also in $R(K')$ and there are even less options available.
    Thus, the coalition's vote count with $R(K')$ is at least as high as with $R(K)$.

    On the other hand, assume that $v$'s favorite party under $R(K)$ is in the opposition.
    If $\topvRP{v}{R(K)} \in R(K')$ - then, as before, it will remain its favorite party under $R(K')$.
    Otherwise, $\topvRP{v}{R(K)} \notin R(K)$. 
    We show that, in this case, voter $v$ prefers either $o_L$ or $o_R$ (or both) over all the coalition parties.
    This will imply that its top party under $R(K')$ must be in the opposition as well.
    
    Let $c\in C $ be an coalition party.
    Clearly, $c \notin I^O$ as all parties in $I^O$ are opposition parties. 
    Let $I^O = [a,b]$. 
    There are only two cases, either $c < a$ or $b < c$.
    If $c<a$, then $c<a\leq \topvRP{v}{R(K)}$. Since $\topvRP{v}{R(K)} \succ_v c$, by \Cref{lemma:divider-meaning}, we can conclude that $\frac{1}{2} (c+\topvRP{v}{R(K)}) < v$. 
    However, since $o_L \leq \topvRP{v}{R(K)}$, then $\frac{1}{2} (c+o_L) < \frac{1}{2} (c+\topvRP{v}{R(K)}) < v$. As $c < o_L$, this means that $o_L \succ_v c$. 
    
    Otherwise, $b < c$, then $\topvRP{v}{R(K)}\leq b< c$.
    As $\topvRP{v}{R(K)} \succ_v c$, \Cref{lemma:divider-meaning} implies that $v < \frac{1}{2} (\topvRP{v}{R(K)} + c)$. 
    However, since $\topvRP{v}{R(K)} \leq o_R$, then $v < \frac{1}{2} (\topvRP{v}{R(K)} + c) < \frac{1}{2} (o_R + c) $. As $o_R < c$, this means that $o_R \succ_v c$.

    Combining these two claims, we establish the desired equality.  
\end{proof}

\paragraph{Control for Coalition and Favored Party.}
Contiguous: \eden{text}

\begin{reminder}{\ref{thm:aop-ssp-CFP-con}}
     For SSP preferences, control for coalition and favored party by adding opposition parties is solvable in polynomial time for contiguous coalition (\withoutTwithLaop).
\end{reminder}

\begin{proof}
    Let $I^O_0, I^O_1$ be the two opposition-intervals, where $I^O_0$ is on the left of the (contiguous) coalition-interval $I_1$ and $I^O_1$ is on its right. 

    We consider any subset $S' \subseteq S$ with $|S'| \leq k$, $|S' \cap I^O_0 | \leq 1$ and $|S' \cap I^O_1 | \leq 1$ -- that is, at most one party is added from the left of the coalition and at most one from the right. 
    For each subset, we check if it satisfies the objectives. 
    If it does, the algorithm stops and asserts "yes". 
    If no subset satisfies the objectives, the algorithm asserts "no".  


    It is clear that this process takes polynomial time, as the number of subsets is less than $m^2$ (where $m$ is the number os parties).
    We shall now see that the algorithm is always correct.
    
    When the algorithm asserts "yes", there indeed exists a subset of size at most $k$ that satisfies the objectives as the algorithm asserts "yes" only when it finds one.  

    When the algorithm returns "no", it means that no subsets $S'$ with $|S'| \leq k$, $|S' \cap I^O_0 | \leq 1$ and $|S' \cap I^O_1 | \leq 1$ satisfy the objectives.
    Assume for contradiction that there exists a subset $K$ with $|K| \leq k$, but either $|K \cap I^O_0 | > 1$ or $|K \cap I^O_1 | > 1$, that satisfies the objectives. 
    
    If $|K \cap I^O_0 | > 1$, let $o_R$ be the right-most party in $K \cap I^O_0$ -- that is, $o < o_R$ for any $o \in K \cap I^O_0$.
    Similarly, if $|K \cap I^O_1 | > 1$, let $o_L$ be the left-most party in $K \cap I^O_1$ -- that is, $o_L < o$ for any $o \in K \cap I^O_0$.
    Notice that at least one of $o_R$ and $o_L$ must exist. Let $K'$ include whichever of them that exists (or both, if they both exist).

    Let $R(K)$ and $R(K')$ be the sets of running parties  when adding the parties in $K$ and $K'$, respectively.
    Let $c \in C$ be a coalition party.
    We prove that the number of votes $c$ gets under $R(K)$ equals to the number of votes it gets under $R(K')$.
    This would imply that if $K$ satisfies the objectives, $K'$ also does. 
    
    To do so, we prove that any voter who votes for $c$ under $R(K)$ still votes for it under $R(K')$; while any voter who votes for the opposition under $R(K)$ still votes for the opposition under $R(K')$.

    Let $v \in V$ be a voter. Assume that $v$'s favorite party under $R(K)$ is $c$. It is clear that $c$ will remain its favorite party under $R(K')$ as the change between $R(K)$ and $R(K')$ involves only opposition parties.
    Thus $c$ is also in $R(K')$ and there are even less options available.

    On the other hand, assume that $v$'s favorite party under $R(K)$ is in the opposition.
    If $\topvRP{v}{R(K)} \in R(K')$ - then, as before, it will remain its favorite party under $R(K')$.
    Otherwise, $\topvRP{v}{R(K)} \notin R(K)$. 
    We show that, in this case, voter $v$ prefers either $o_L$ or $o_R$ (or both) over all the coalition parties.
    This will imply that its top party under $R(K')$ must be in the opposition as well.

    Let $c'\in C $ be an coalition party.
    If $\topvRP{v}{R(K)} \in I^O_0$, then $\topvRP{v}{R(K)} < o_R < c'$.
    Since $\topvRP{v}{R(K)} \succ_v c'$, by \Cref{lemma:divider-meaning}, we can conclude that $v < \frac{1}{2} (\topvRP{v}{R(K)} + c')$. 
    However, since $\topvRP{v}{R(K)} < o_R$, then $ v < \frac{1}{2} (\topvRP{v}{R(K)} + c') < \frac{1}{2} (o_R + c') $. As $o_R < c'$, this means that $o_R \succ_v c'$. 
    
    Otherwise, $\topvRP{v}{R(K)} \in I^O_1$, then $ c' < o_L < \topvRP{v}{R(K)}$.
    Since $\topvRP{v}{R(K)} \succ_v c'$, by \Cref{lemma:divider-meaning}, we can conclude that $ \frac{1}{2} (c' + \topvRP{v}{R(K)} ) < v$. 
    However, since $o_L < \topvRP{v}{R(K)}$, then $\frac{1}{2} (c' + o_L ) < \frac{1}{2} (c' + \topvRP{v}{R(K)} ) < v$. As $c' < o_L$, this means that $o_L \succ_v c'$.

    Thus, $K'$ satisfies both objectives.
    However, since $|K' \cap I^O_0 | \leq 1$ and $|K' \cap I^O_1 | \leq 1$, the algorithm must have considered and checked $K'$. 
    But the algorithm asserted "no", which means that $K'$ did not satisfy the objectives -- a contradiction.  
\end{proof}

\begin{reminder}{\ref{thm:CCFP-AOP-pseudo-polynomial}}
    For SSP preferences Control for coalition and favored party by adding opposition parties admits a pseudo-polynomial algorithm for the compact representation.
\end{reminder}

\begin{proof}
    The algorithm is similar to the one presented in \Cref{thm:CCFP-AOP-pseudo-polynomial}.

    As before, we define the set $A$ of acceptable vote combinations, representing pairs of votes for the coalition and the favored party that meet the objectives, as follows:
    \begin{align*}
    A := \left\{ (n_c, n_f) \in \{1, \ldots, n\}^2\ \middle\vert \begin{array}{l}
        \frac{1}{n}(n_f + n_c) \geq \varphi\\
        n_f \geq \rho \cdot n_c
      \end{array}\right\}
    \end{align*}

    The goal is to find out whether such vote combination can be exactly obtained by adding at most $k$ opposition parties from $S$. 
    
    The algorithm is based on two key observations.
    \Cref{lem:sperate-opposition-intervals} ensures that different opposition-intervals can be treated separately; while \Cref{lemma:adding-opposition-two}, ensures that it suffices to consider only subsets of size at most $2$ in each interval.

    \eden{todo: to copy last pseudo version and edit}
\end{proof}

\subsection{Deleting Coalition Parties}
\begin{reminder}{\ref{thm:dcp-ssp-CFP-con}}
For SSP preferences, control for coalition and favored party by deleting coalition parties polynomial-time solvable for contiguous coalitions (\withoutTwithLdcp).
\end{reminder}
\begin{proof}
   The algorithm evaluates whether specific vote combinations meet the objectives.

Formally, we define the set $A$ of acceptable vote combinations, $(n_c, n_f) \in A$, representing pairs of number of votes shifted from the  coalition to the favored party ($n_f$) and the number of votes shifted from the coalition to the opposition ($n_c$) that meet the objectives.
That is:
    \begin{align*}
    A^- := \left\{ (n_c, n_f) \in \{1, \ldots, n\}^2\ \middle\vert \begin{array}{l}
        \frac{(N(P, C) - n_c)}{n} \geq \varphi\\
        \frac{N(P, p_1) + n_f}{(N(P, C \setminus\{p_1\}) - n_c)} \geq \rho 
      \end{array}\right\}
    \end{align*}
The goal is to find out whether such votes combination can be exactly obtained by deleting at most $k$ parties.
    
In the contiguous coalition case, the problem can be solved in polynomial time. Since there is only one coalition-interval, we can examine all possible ways to partition the $k$ deleted parties between the two endpoints of that coalition-interval and the two sides of the favored party. Then check if one of the optional outcomes is in $A^-$.
\end{proof}

\begin{reminder}{\ref{thm:CCFP-dcp-pseudo-polynomial}}
    For SSP preferences,
    Control for coalition and favored party by deleting coalition parties admits a pseudo-polynomial algorithm for the compact representation
  (\withoutTwithLdcp).
\end{reminder}
\begin{proof}
    As in the proof of \Cref{thm:dcp-ssp-CFP-con},  we define the set $A$ of acceptable vote combinations, $(n_c, n_f) \in A$, representing pairs of number of votes shifted from the coalition to the favored party ($n_f$) and the number of votes shifted from the coalition to the opposition ($n_c$) that meet the objectives.
    That is:
    \begin{align*}
    A^- := \left\{ (n_c, n_f) \in \{1, \ldots, n\}^2\ \middle\vert \begin{array}{l}
        \frac{(N(P, C) - n_c)}{n} \geq \varphi\\
        \frac{N(P, p_1) + n_f}{(N(P, C \setminus\{p_1\}) - n_c)} \geq \rho 
      \end{array}\right\}
    \end{align*}
    
    \paragraph{Algorithm Description.} 
     \paragraph{(1) Per-Interval Calculation.} 
    For each $q' = 0,\ldots,q+1$, ~~$i = 0,1,\ldots, k$,~~ $n'_c = 0, \ldots, n$,  and $n'_f = 0, \ldots, n$; we check whether it is possible to move exactly $n_c$ votes from the coalition to the opposition  and achieve $n_f$ votes for the favored party, by deleting at most $i$ parties from the interval $I^O_{q'}$.
    To compute this, we use the same technique from \Cref{thm:dcp-ssp-CFP-con}.
    The result, a "Yes" or "No" answer, is denoted by $W'(I_{q'},i,n'_c, n'_f)$.

    \paragraph{(2) Dynamic Programming.} 
    Let $W(\mathcal{I}_{q'},k',n_c, n_f)$ be the answer to the same question but by deleting at most $k'$ parties from the intervals $\mathcal{I}_{q'}:= \{I_1,\ldots, I_{q'}\}$, as follows. 

    \begin{itemize}
        \item Base Case: $W(\mathcal{I}_{1},k',n_c, n_f) =  W'(\mathcal{I}_{1},k',n_c, n_f)$.
         
        \item Recursive Formula for the Dynamic Programming: 
        $W(\mathcal{I}_{q'},k',n_c, n_f) = "yes"$ if and only if exist $k'',n''_c, n''_f$
          such that $W'(\mathcal{I}_{q'},k'',n'_c, n'_f) = "yes"$ and $W(\mathcal{I}_{q'-1},k'-k'',n_c-n'_c, n_f-n'_f) = "yes"$.
       
    \end{itemize}
    \paragraph{(3) Objective Check.} We check, for $k'=0,\ldots k$ and any acceptable combination $(n'_c, n'_f) \in A^-$, whether the answer in $W(\mathcal{I}_{q},k',n'_c, n'_f)$ is "yes".
    If so, then the algorithm asserts "yes"; otherwise it asserts "no".

    Notice that the complexity is polynomial in the extensive representation, but only pseudo polynomial in the compact representation (where the input size does not depend on $n$).
  
\end{proof}

\subsection{Deleting Opposition Parties}
\begin{reminder}{\ref{thm:dop-ssp-C}}
For SSP preferences, control for a coalition by deleting opposition parties is polynomial-time solvable (\withoutTwithoutLdop).
\end{reminder}
\begin{proof}
    We present a dynamic programming algorithm to determine the maximum votes a coalition can achieve by deleting up to $k$ opposition parties, allowing us to verify if the coalition's objective is attainable.

    The algorithm relies on two observations: (1) By \Cref{lem:sperate-opposition-intervals}, deleting parties from one opposition-interval does not affect the impact of deleting an opposition party from another opposition-interval. (2) By \Cref{lemma:most-adjacent-only}, to maximize votes that move to the coalition by deleting $k'$   parties from one opposition-interval, it suffices to consider their distribution between the two sides of that opposition-interval.
    Combining these results, we restrict the set of combinations to be evaluated.
    \paragraph{Dynamic Programming Algorithm Description.}
    \paragraph{(1) Per-Interval Calculation.} 
    we start by calculating for each $q' = 0,\ldots,q+1$ and for each $i = 0, \ldots, k$, the maximum number of votes that can be added to the coalition by deleting $i$ parties from the interval $I^O_{q'}$, denoted by $N'_C(I^O_{q'}, i)$.
    The computation of $N'_C(I^O_{q'}, i)$ involves checking all ways to distribute these $i$ deletions between the two sides, overall $i+1$ polynomial computations. Thus, it can be done in polynomial time.

    Note that $N'_C$ is the number of added votes, which refers to the votes added to the coalition due to the deletion of opposition parties. Thus, the total number of votes for the coalition is the original number of votes the coalition received before the control action, plus $N'_C$.


    \paragraph{(2) Recursive Formula for the Dynamic Programming .} 
    For each $q' = 0,\ldots,q+1$ and for $k' = 0,\ldots, k$, we calculate the maximum number of votes that can be added to the coalition by deleting at most $k'$ parties from the intervals $\mathcal{I}_{q'}:= \{I_0,\ldots, I_{q'}\}$, denoted by $N_C(\mathcal{I}_{q'},k')$.
    \begin{itemize}
        \item Base Cases:  $N_C(\mathcal{I}_{0},k') = N'_C(I^O_{0}, k')$.

        \item Recursive Step, for $q' \geq 1$:
        \begin{align*}
            N_C(\mathcal{I}_{q'},k'):= \displaystyle \max_{0\leq i \leq k'}(&N'_C(I_{q'}, i) \\
            &+ N_C(\mathcal{I}_{q'-1},k' -i) )
        \end{align*}
    \end{itemize}

    \paragraph{(3) Objective Check.} 
    We check whether the coalition objective is satisfied when the number of added votes to the coalition is $N_C(\mathcal{I}_{q+1},k)$ -- that is, whether:
        \begin{align*}
            \frac{N(P, C) + N_C(\mathcal{I}_{q+1},k)}{n} \geq \varphi
        \end{align*}
    If so, then the algorithm asserts "yes"; otherwise it asserts "no".
    
\end{proof}

\begin{reminder} {\ref{thm:dop-ssp-CFP-con}}
    For SSP preferences, control for coalition and favored party by deleting opposition parties is polynomial-time solvable  for contiguous coalition (\withoutTwithLdop).
\end{reminder}
\begin{proof}
    There are at most two opposition-intervals, denoted by $I^O_0$ and $I^O_1$.
    We define the set $A$ of acceptable vote combinations, $(n_c, n_f) \in A$, representing pairs of additional votes shifted from the opposition to the coalition ($n_c$) or the favored party ($n_f$) that meet the objectives.
    That is:
    \begin{align*}
    A^+ := \left\{ (n_c, n_f) \in \{1, \ldots, n\}^2\ \middle\vert \begin{array}{l}
        \frac{(N(P, C) + n_f + n_c)}{n} \geq \varphi\\
        \frac{N(P, p_1) + n_f}{(N(P, C \setminus\{p_1\}) + n_c)} \geq \rho 
      \end{array}\right\}
    \end{align*}
    The goal is to find out whether such votes combination can be exactly obtained by deleting at most $k$ parties.

    First, by \Cref{lem:sperate-opposition-intervals}, deleting parties from one opposition-interval does not affect the impact of deleting an opposition party from another opposition-interval. Second, by \Cref{lemma:most-adjacent-only}, given that we want to delete $k'$  parties from one opposition-interval, it suffices to consider how to distribute these $k'$ deletions between the two sides of that opposition-interval, since deleting other opposition party has no effect.
    Note that we will also check the option of deleting fewer than $k$ opposition parties, as we may not want a coalition that is too large. 
       
    Hence, we consider all possible ways to split the $k$ deleted parties between $I^O_0$ and $I^O_1$ (there are $k+1$ options). 
    For each $t\in\{0,1\}$, for each $0\leq k' \leq k$ we compute all possible pairs of addition votes to the favored party, and addition votes to the other parties in the coalition, let $N_{F, C}(I_{t},k')$ be the list of all such pairs (there are $k'+1$ such pairs).

    If there exists a couple of pairs, $(n^0_c, n^0_f) \in N_{F, C}(\mathcal{I}_{0},k_0)$, $(n^1_c, n^1_f) \in N_{F, C}(\mathcal{I}_{1},k_1)$ such that:
      
    (1) $k_0 + k_1 \leq k$ and (2) $(n^0_c + n^1_c, n^0_f + n^1_f) \in A$
    
    There are only a polynomial number of couples of pairs to check if they satisfy both objectives, ensuring that the total runtime of the algorithm remains polynomial.

\end{proof}

\begin{reminder}{\ref{thm:CCFP-DOP-pseudo-polynomial}}
    For SSP preferences, control for coalition and favored party by deleting opposition parties admits a pseudo-polynomial algorithm for the compact representation.
\end{reminder}
\begin{proof}
    The algorithm is conceptually similar to the one presented in \Cref{thm:dop-ssp-C} but with a key difference: instead of computing the maximum number of votes attainable by the coalition, it evaluates whether specific vote combinations meet the objectives.

   As in the proof of \Cref{thm:dop-ssp-CFP-con}, we define the set $A$ of acceptable vote combinations, $(n_c, n_f) \in A$, representing pairs of additional votes shifted from the opposition to the coalition ($n_c$) or the favored party ($n_f$) that meet the objectives.
    That is:
    \begin{align*}
    A^+ := \left\{ (n_c, n_f) \in \{1, \ldots, n\}^2\ \middle\vert \begin{array}{l}
        \frac{(N(P, C) + n_f + n_c)}{n} \geq \varphi\\
        \frac{N(P, p_1) + n_f}{(N(P, C \setminus\{p_1\}) + n_c)} \geq \rho 
      \end{array}\right\}
    \end{align*}
    The goal is to find out whether such votes combination can be exactly obtained by deleting at most $k$ parties. 

    As before, the algorithm is based on two key observations. First, by \Cref{lem:sperate-opposition-intervals}, deleting parties from one opposition-interval does not affect the impact of deleting an opposition party from another opposition-interval. Second, by \Cref{lemma:most-adjacent-only}, given that we want to delete $k'$  parties from one opposition-interval, it suffices to consider how to distribute these $k'$ deletions between the two sides of that opposition-interval, since deleting other opposition party has no effect.
    Note that we will also check the option of deleting fewer than $k$ opposition parties, as we may not want a coalition that is too large.

    \paragraph{Algorithm Description.} 
     \paragraph{(1) Per-Interval Calculation.} 
    For each $q' = 0,\ldots,q+1$, ~~$i = 0,1,\ldots, k$,~~ $n'_c = 0, \ldots, n$,  and $n'_f = 0, \ldots, n$; we check whether it is possible to achieve exactly $n_c$ votes for the coalition and $n_f$ votes for the favored party, by deleting at most $i$ parties from the interval $I^O_{q'}$.
    To compute this, we use the same technique of checking all options for distributing these $i$ deletions between the two sides of the opposition-interval, as proved in \Cref{lemma:most-adjacent-only}.
    The result, a "Yes" or "No" answer, is denoted by $W'(I_{q'},i,n'_c, n'_f)$.
  
   \paragraph{(2) Dynamic Programming.}
    Let $W(\mathcal{I}_{q'},k',n_c, n_f)$ be the answer to the question whether it is possible to achieve exactly $n_c$ votes for the coalition and $n_f$ votes for the favored party by deleting at most $k'$ parties from the intervals $\mathcal{I}_{q'}:= \{I_1,\ldots, I_{q'}\}$, as follows.

    \begin{itemize}
        \item Base Case: $W(\mathcal{I}_{1},k',n_c, n_f) =  W'(\mathcal{I}_{1},k',n_c, n_f)$.
         
        \item Recursive Formula for the Dynamic Programming: 
        $W(\mathcal{I}_{q'},k',n_c, n_f) = "yes"$ if and only if exist $k'',n''_c, n''_f$
          such that $W'(\mathcal{I}_{q'},k'',n'_c, n'_f) = "yes"$ and $W(\mathcal{I}_{q'-1},k'-k'',n_c-n'_c, n_f-n'_f) = "yes"$.
       
    \end{itemize}
    \paragraph{(3) Objective Check.} We check, for $k'=0,\ldots k$ and any acceptable vote combination $(n'_c, n'_f) \in A$, whether the answer in $W(\mathcal{I}_{q},k',n'_c, n'_f)$ is "yes".
    If so, then the algorithm asserts "yes"; otherwise it asserts "no".

    Notice that the complexity is polynomial in the extensive representation, but only pseudo polynomial in the compact representation (where the input size does not depend on $n$).

\end{proof}

\section{Missing Immune Proofs}\label{apx:immune}
\subsection{Proof of Theorem \ref{thm:-ccdcp-immune}}

\begin{reminder}{\ref{thm:-ccdcp-immune}}
     Control for the coalition by deleting coalition parties (\withoutTwithoutLdcp) is not possible (the problem is immune).
\end{reminder} 

\begin{proof}
    Each voter gives exactly one point to the party ranked first.  
    Since all deleted parties are from the coalition, we distinguish between two cases:
    \begin{enumerate}
        \item  For a voter whose first-ranked party is not from the coalition, deleting a coalition party, does not change the outcome, and the total points received by opposition parties remain the same.
        \item For a voter whose first-ranked party is from the coalition, deleting a coalition party that is in the first position may reduce the points received by the coalition, if the next party is from the opposition.  
    \end{enumerate}
    
    This problem concerns only the percentage of points received by the coalition as a whole, rather than by any specific favored party within it.
    As a result, deleting opposition parties can only harm the coalition's overall standing.
\end{proof}

\section{Missing Hardness Proofs} \label{appendix:sec:hard}
For the hardness proofs, we need the following problem definitions: 
\begin{definition}[Dominating Set Problem]
    \textbf{Given:} A graph $G = (W,E)$, and a parameter $k$.
    \textbf{Question:} Does there exists a dominating set $W'\subseteq W$ of size at most $k$. That is, $\forall w\in W: w\in W' \lor \exists w'\in W'\ s.t\ (w,w')\in E$.
\end{definition}
The Dominating Set Problem is w[2]-hard \cite{downey1999structure}.

\begin{definition}[k-Subset Sum Problem]
    \textbf{Given:} A set $A$ of positive integers and a target integer $\tau$, parameter $k$.
    
    \textbf{Question:} Does there exist a subset $A' \subseteq A$, such that $|A'| \leq k$, and $\sum_{a\in A'}a = \tau$.
\end{definition}
The k-subset Sum Problem is an NP-hard problem, even when the given set $A$ contains only positive integers. 

\begin{definition}[k-clique problem]
    \textbf{Given:} A graph $G = (W,E)$, and a parameter $k$.
    \textbf{Question:} Does there exists a clique of size $k$ in the graph. That is a set, $W'$, of $k$ vertices such that $\forall w,w'\in W' (w,w')\in E$.
\end{definition}
The k-clique problem is a w[1]-hard problem \cite{Fixed1995Downey}.

\subsection{Proof of Theorem~\ref{thm:acp-general-C}}\label{apx:acp-general-C}
The proof is by a reduction from the dominating set problem. 
\begin{reminder}{\ref{thm:acp-general-C}}
    For general preferences, control for the coalition by adding coalition parties (\withoutTwithoutLacp) is W[2]-hard.
\end{reminder}

\begin{proof}
    Given a graph $G=(W, E)$, where $|W| = n$, $|E| = m$ and a parameter $k$ of the dominating set instance.  
    Construct an instance of the \withoutTwithoutLacp\ problem with parameter $k$ such that a successful control action exists if and only if there exists a dominating set.

    \textbf{Construction.}
    Let $P= \{p_1, p_2\}$ be the set of parties.  
    Define the set of spoiler parties as $S=\{s_1,\ldots, s_n\}$, contains a spoiler party for each vertex in $G$, that is, each spoiler party $s_i$ corresponds to the vertex $w_i\in W$ of $G$.
    Let the coalition be $C=\{p_1,s_1,\ldots, s_n\}$.
    Let the parameter of the number of added parties be $k$.
    Let the target-fraction $\varphi = \frac{1}{2}$.

    \textbf{Voter Preference Order:}
    Let $V = \{v_1, \ldots, v_n, v_{n+1}, \ldots, v_{2n}\}$ be the set of voters, each voter is corresponding to one vertex. 
    Let $N[s_i]$ be the set of parties that contain the party $s_i$ and all $s_j\in S$ such that $(w_i,w_j)\in D$.
    For $1\leq i \leq n$, the preference order of $v_i$ is defined to be $N[s_i] \succ p_2 \succ p_1 \succ S\setminus N[s_i]$. 
    For $n+1\leq i \leq 2n$ , the preference order of $v_i$ is defined to be $p_2 \succ p_1 \succ S$

    \textbf{To Sum Up.}
    Assume there exists a dominating set of size at most $k$.
    In this case, the control action that adds the corresponding spoiler parties ensures that the coalition will achieve $50\%$ of the votes.
    This is because, for $1\leq i \leq n$ the first party in each preference order of voter $v_i$ is one of the spoiler parties that belong to the coalition and is part of the dominating set.
    For $n+1\leq i \leq 2n$, the first party in each preference order of voter $v_i$ is $p_2$, which is not part of the coalition, and it cannot be changed by adding spoiler parties.
    As a result, the coalition achieves $50\%$ of the votes.

    On the other hand, assume that there is successful control action, $S'\subseteq S$, that added at most $k$ spoiler parties and achieves $50\%$ of the votes for the coalition. 
    For $n+1\leq i \leq 2n$, the first party in each preference order of voter $v_i$ is $p_2$, which is not part of the coalition, and it cannot be changed by adding parties.
    Hence, it must be that for all $1\leq i \leq n$, the first party in each preference order of voter $v_i$ is part of the coalition. 
    That is, in each voter one of $N[s_i]$ must be part of $S'$, and $|S'|\leq k$, hence the set that contains the corresponding vertices is a dominating set of the given graph.

    To conclude, we proved that \withoutTwithoutLacp\ is w[2]-hard.
\end{proof}

\subsection{Proof of Theorem \ref{thm:acp-ssp-CFP-non-con}}
The proof is by a reduction from the k-subset sum problem.
\begin{reminder}{\ref{thm:acp-ssp-CFP-non-con}}
    For SSP preferences, control for the coalition and favored party by adding coalition parties (\withoutTwithLacp) is NP-hard under the compact representation.
\end{reminder}
\begin{proof}
Given an instance of the $k$-subset sum problem, let $A = \{a_1, a_2, \ldots, a_m\}$ be the set of integers, $\tau$ be the target integer, and $k$ be the parameter. We construct an instance of the \withoutTwithLacp\ problem such that there is a solution to the $k$-subset sum problem if and only if there exists a successful control action.

Let the set of parties be $P = \{p_1, b_1, b_2, \ldots, b_m\}$ and the set of spoiler parties be $S = \{s_1, s_2, \ldots, s_m\}$.
The coalition is $C = \{p_1, s_1, s_2, \ldots, s_m\}$, and the favored party is $p_1$.
That is, each integer $a_i$, has a corresponding opposition-party $b_i$, and a corresponding spoiler coalition party $s_i$.

Place the parties on the interval such that the closest party to each $s_i$ is $b_i$.

Let the set of voters $V$ include:
\begin{itemize}
    \item For each $a_i \in A$, voters $v_{i,1}, v_{i,2}, \ldots, v_{i,a_i}$ whose peak is at the spoiler party $s_i$, that is $h(v_{i,j}) = s_i$.
    \item $\tau$ voters $v_1, v_2, \ldots, v_{\tau}$ whose peak is at the favored party $p_1$, that is $h(v_i) = p_1$.
\end{itemize}

Define the parameter for the number of added parties as $k$.

Let the target fraction be $\varphi = \frac{2 \cdot \tau}{\tau + \sum_{a \in A} a}$ and the target ratio be $\rho = \frac{1}{2}$.

\paragraph{Summary.}
All voters of a spoiler party that does not run will move to the corresponding opposition party. Thus, before adding any spoiler party, the opposition receives $\sum_{a \in A} a$ votes, and the favored party receives $\tau$ votes.

Assume there is a subset $A' \subseteq A$ such that $|A'| \leq k$ and $\sum_{a \in A'} a = \tau$. Adding the corresponding spoiler parties results in exactly $\tau$ new votes for the coalition, satisfying both requirements.

On the other hand, assume there is a successful control action. This implies there is a set of at most $k$ spoiler parties whose addition meets both requirements.

Since the requirement of $\varphi$ holds, the coalition must receive at least $2 \cdot \tau$ votes from the total $\tau + \sum_{a \in A} a $ votes.

In addition, since the favored party receives only $\tau$ votes and cannot gain additional votes through control, the other coalition members can collectively achieve at most $\tau$ votes.

Therefore, the coalition receives exactly $\tau$ votes.

Adding any spoiler party $s_i$ contributes exactly $a_i$ votes to the coalition. Hence, the subset containing elements corresponding to the spoiler parties added in control satisfies the requirements and forms a solution to the $k$-subset sum problem.
\end{proof}

\subsection{Proof of Theorem~\ref{thm:aop-general-CFP}}
    
The proof is by reduction from the dominating set problem.
\begin{reminder}{\ref{thm:aop-general-CFP}}
    For general preferences, control for the coalition and favored party by adding opposition parties (\withoutTwithLaop) is W[2]-hard.
\end{reminder}
\begin{proof}
    Given a graph $G=(W, D)$, where $|W| = n$, $|E| = m$ and parameter $k$ of the dominating set instance. Construct an instance of the \withoutTwithLaop\ problem with parameter $k$ such that there is a successful control action if and only if exists a dominating set.
    Let $P= \{p_1,p_2\}$ be the set of parties.  
    Let the set of spoiler parties $S=\{s_1,\ldots, s_n\}$, contain a spoiler party for each vertex in $G$.
    Let the coalition be $C=\{p_1,p_2\}$, and the favored party be $p_1$.
    Let the parameter of the number of added parties be $k$.
    Let $\varphi = 50\%$, and let $\rho = 50\%$.

    \textbf{Voter Preference Order:}
    Let $V = \{v_1, \ldots, v_n, v_{n+1}, \ldots, v_{2n}\}$ be the set of voters.
    Let $N[s_i]$ be the set of parties that contain the party $s_i$ and all $s_j\in S$ such that $(w_i,w_j)\in D$.
    For $1\leq i \leq 0.5n$, the preference order of $v_i$ is defined to be $p_1 \succ p_2 \succ S$
    For $0.5n \leq i \leq n$, the preference order of $v_i$ is defined to be $p_2 \succ p_1 \succ S$
    For $n+1\leq i \leq 2n$, each voter corresponds to one vertex, the preference order of $v_i$ is defined to be $N[s_i] \succ p_2 \succ p_1 \succ S\setminus N[s_i]$. 
 
    Before any control action, the coalition achieves $100\%$ off the votes, but, $p_1$ achieves only $0.25\%$ of the coalition. 
    \textbf{To Sum Up.}
    Assume there exists a dominating set of size at most $k$.
    In this case, the control action that adds the corresponding spoiler parties ensures that the coalition will achieve only $50\%$ of the votes.
    This is because, for $1\leq i \leq n$ the first party in each preference order of voter $v_i$ is not one of the spoiler parties and it cannot be changed by control by adding parties.
    After the control action, for $n+1\leq i \leq 2n$, the first party in each preference order of voter $v_i$ is a spoiler party that is not part of the coalition. 
    As a result, the coalition achieves $50\%$ of the votes, and the percentage of $p_1$ among the coalition becomes $0.5\%$ as demanded.

    On the other hand, assume that there is a successful control action, $S'\subseteq S$, that added at most $k$ spoiler parties and achieves a successful control action.
    It is not possible to achieve more votes to $p_1$, so the only way to meet the requirement of $\rho$ is by reducing the size of the coalition, without reducing the number of votes to $p_1$.
    For $1\leq i \leq n$, the first party in each preference order of voter $v_i$ cannot be changed by control by adding parties.

    Hence, it must be that for $n+1\leq i \leq 2n$ the first party in each preference order of voter $v_i$ is not part of the coalition after the control action. 
    That is, in each voter one of $N[s_i]$ must be part of $S'$, and $|S'|\leq k$, hence the set that contains the corresponding vertices is a dominating set of the given graph.
\end{proof}

\subsection{Proof of Theorem~\ref{thm:aop-ssp-CFP-non-con}}
The proof is by a reduction from the k-subset sum problem.

\begin{reminder}{\ref{thm:aop-ssp-CFP-non-con}}
   For SSP preferences, control for the coalition and favored party by adding opposition parties (\withoutTwithLaop) is NP-hard under the compact representation.
\end{reminder}
\begin{proof}
    Given an instance of the $k$-subset sum problem, let $A = \{a_1, a_2, \ldots, a_m\}$ be the set of integers, $\tau$ be the target integer, and $k$ be the parameter. We construct an instance of the \withoutTwithLaop\ problem such that there is a solution to the $k$-subset sum problem if and only if there exists a successful control action.

    Let the set of parties be $P = \{p_1, b_1, b_2, \ldots, b_m\}$ and the set of spoiler parties be $S = \{s_1, s_2, \ldots, s_m\}$.
    The coalition is $C = \{p_1, b_1, b_2, \ldots, b_m\}$, and the favored party $p_1$.

    That is, each integer $a_i$, has a corresponding opposition-spoiler-party $s_i$, and a corresponding coalition party $b_i$.

    Place the parties on the interval such that the closest party to each $s_i$ is $b_i$.
    Let $\bar{\tau} = \sum_{a \in A'} a - \tau$
    Let the set of voters $V$ include:
    \begin{itemize}
        \item For each $a_i \in A$, voters $v_{i,1}, v_{i,2}, \ldots, v_{i,a_i}$ whose peak is at the spoiler party $s_i$, that is $h(v_{i,j}) = s_i$.
        \item $\bar{\tau}$ voters $v_1, v_2, \ldots, v_{\bar{\tau}}$ whose peak is at the favored party $p_1$, that is $h(v_i) = p_1$.
    \end{itemize}

    Define the parameter for the number of added parties as $k$. 

    Let the target fraction be $\varphi = \frac{2 \cdot \bar{\tau}}{\bar{\tau} + \sum_{a \in A} a}$ and the target ratio be $\rho = \frac{1}{2}$.

    \paragraph{Summary.}
    All voters of a spoiler party that does not run will move to the corresponding coalition party. Thus, before adding any spoiler party, the coalition without the favored party receives $\sum_{a \in A} a$ votes, and the favored party receives $\bar{\tau}$ votes.

    Assume there is a subset $A' \subseteq A$ such that $|A'| \leq k$ and $\sum_{a \in A'} a = \tau$. Adding the corresponding spoiler parties results in exactly $\tau$ votes moving from the coalition to the opposition, satisfying both requirements, since $p_1$ receives $\bar{\tau}$ votes and the coalition without the favored party receives also $\bar{\tau}$ votes.

    On the other hand, assume that there is a successful control action. This implies that there is a set of at most $k$ spoiler parties whose addition meets both requirements.

    (1) Since the requirement of $\varphi$ holds, the coalition must receive at least $2 \cdot \bar{\tau}$ votes from the total $\bar{\tau} + \sum_{a \in A} a$ votes.

    (2) Since the favored party receives only $\bar{\tau}$ votes and cannot gain additional votes through control, the other coalition members can collectively achieve at most $\bar{\tau}$ votes.

    Therefore, the coalition receives exactly $\bar{\tau}$ votes.

    Adding any spoiler party $s_i$ removes exactly $a_i$ votes from the coalition. Hence, the subset containing elements corresponding to the spoiler parties added in control satisfies the requirements and forms a solution to the $k$-subset sum problem.
\end{proof}

\subsection{Proof of Theorem \ref{thm:dcp-general-CFP}}
The proof is by a reduction from the k-clique problem.
\begin{reminder}{\ref{thm:dcp-general-CFP}}
    For general preferences, control for the coalition and favored party by deleting coalition parties (\withoutTwithLdcp) is W[1]-hard.
\end{reminder}
\begin{proof}
    Given an instance of the k-clique problem. Let the graph be $G=(W, D)$, where $|W| = n$, $|E| = m$, and parameter $k$, where the parameter is the size of the clique. 
    Construct an instance of the \withoutTwithLdcp\ problem with parameter $k$ such that there is a successful control action if and only if a k-clique exists.

    Let $P=\{p_1, p_2, s_1,\ldots, s_n\}$, each $s_i$ is corresponding to a vertex in $G$.
    Let the coalition to be $C=\{p_1, s_1,\ldots, s_n\}$, and the preferred party $p_1$. 
    Let the parameter of the number of deleting parties be $k$.
    Let the target fraction $\varphi = \frac{1}{2}$, and the target ratio $\rho = \frac{m+k^2}{2m}$.

    \textbf{Voter Preference Order:}
    Let $V = \{v_{1,1},\ldots v_{1,m}, v_{2,1},\ldots v_{2,m}, v_{3,1},\ldots v_{3,m}, v_{4,1},\ldots v_{4,m}\}$  be the set of voters, $|V| = 4m$.
    Each voter $v_{1, i}$ is corresponding to one edge in $G$. 
    The preference order of $v_{1,i}$, that corresponding to the edge $(w_x,w_y)$ is defined to be $s_x \succ s_y \succ p_1 \succ p_2 \succ P\setminus \{s_x,s_y,p_1, p_2\}$.
    The preference order of $v_{2,i}$ is defined to be $p_1 \succ p_2 \succ P\setminus \{p_1,p_2\}$.
    The preference order of $v_{3,i}$ and $v_{4,i}$ is defined to be $p_2 \succ P\setminus \{p_2\}$.
    
    \textbf{To Sum Up.}
    
    Assume there exists a k-clique in the graph, then delete the corresponding $k$ parties.
    Each edge between these vertices is in the graph, hence there are exactly $k^2$ such edges. Thus, in the preference order of all voters that correspond to these edges, after the deletion, the first party will be $p_1$, so $p_1$ will get $m+k^2$ votes.
    In the other $m-k^2$ voters $v_{1,i}$ the first party is not $p_1$ but it is  part of the coalition
    In addition, $p_2$, which is not part of the coalition, gets  $2m$ votes, and it cannot be changed by deletion of parties from the coalition.
    Hence, the coalition achieve $\varphi = \frac{1}{2}$, and $p_1$ achieve $\frac{m+k^2}{2m}$ of the coalition's votes as desired.

    On the other hand, suppose a successful control action exists.  
    In the original election, $p_2$, which is not part of the coalition, gets  $2m$ votes, and it cannot be changed by deletion of parties from the coalition.
    The parties $s_1,\ldots, s_n$, are part of the coalition.
    In order that $p_1$ will achieve the target ratio, $\rho = \frac{1}{2}$, $p_1$, must get at least $m+k^2$ votes, that is, achieve at least $k^2$ new votes by the control action. 

    If a deletion results in a vote for $p_1$, it implies that both vertices corresponding to the endpoints of the associated edge are among the deleted parties.  
    This means the deletion of $k$ parties successfully includes the vertices representing both endpoints of the edges for $k^2$ edges.  

    Consequently, all edges between these vertices must exist in the graph $G$, indicating that these vertices form a clique. 
    
\end{proof}

\subsection{Proof of Theorem \ref{thm:dcp-ssp-CFP-non-con}}
The proof is by a reduction from the $k$-subset sum problem.
\begin{reminder}{\ref{thm:dcp-ssp-CFP-non-con}}
   For SSP preferences, control for the coalition and favored party by deleting coalition parties (\withoutTwithLdcp) is NP-hard  under the compact representation.
\end{reminder}
\begin{proof}
    Given an instance of the $k$-subset sum problem, let $A = \{a_1, a_2, \ldots, a_m\}$ be the set of integers, $\tau$ be the target integer, and $k$ be the parameter. We construct an instance of the \withoutTwithLdcp\ problem such that there is a solution to the $k$-subset sum problem if and only if there exists a successful control action.

    Let the set of parties be $P = \{p_1, b_1, b_2, \ldots, b_m, s_1, s_2, \ldots, s_m\}$.
    The coalition is $C = \{p_1, b_1, b_2, \ldots, b_m\}$, and the the favored party is $p_1$.
    That is, each integer $a_i$, has a corresponding opposition-spoiler-party $s_i$, and a corresponding coalition party $b_i$.

    Place the parties on the interval such that the closest party to each $s_i$ is $b_i$.
     Let $\bar{\tau} = \sum_{a \in A'} a - \tau$.
      Let the set of voters $V$ include:
    \begin{itemize}
        \item For each $a_i \in A$, voters $v_{i,1}, v_{i,2}, \ldots, v_{i,a_i}$ whose peak is at the coalition party $b_i$, that is $h(v_{i,j}) = b_i$.
        \item $\bar{\tau}$ voters $v_1, v_2, \ldots, v_{\bar{\tau}}$ whose peak is at the favored party $p_1$, that is $h(v_i) = p_1$.
    \end{itemize}
    Define the parameter for the number of deleted parties as $k$. 
    Let the target fraction be $\varphi = \frac{2 \cdot \bar{\tau}}{\bar{\tau} + \sum_{a \in A} a}$ and the target ratio be $\rho = \frac{1}{2}$.

    \paragraph{Summary.}
    All voters of a deleted coalition party will move to the corresponding opposition party. Thus, before deleting any coalition party, the coalition without the favored party receives $\sum_{a \in A} a$ votes, and the favored party receives $\bar{\tau}$ votes.

    Assume there is a subset $A' \subseteq A$ such that $|A'| \leq k$ and $\sum_{a \in A'} a = \tau$. Deleting the corresponding coalition parties results in exactly $\tau$ votes moving from the coalition to the opposition, satisfying both requirements, since $p_1$ receives $\bar{\tau}$ votes and the coalition without the favored party receives also $\bar{\tau}$ votes.

    On the other hand, assume there is a successful control action. This implies there is a set of at most $k$ coalition parties whose deletion meets both requirements.

    (1) Since the requirement of $\varphi$ holds, the coalition must receive at least $2 \cdot \bar{\tau}$ votes from the total $\bar{\tau} + \sum_{a \in A} a$ votes.

    (2) Since the favored party receives only $\bar{\tau}$ votes and cannot gain additional votes through control, the other coalition members can collectively achieve at most $\bar{\tau}$ votes.

    Therefore, the coalition without the favored party receives exactly $\bar{\tau}$ votes.

    Deleting any coalition party $b_i$ removes exactly $a_i$ votes from the coalition. Hence, the subset containing elements corresponding to the deleted coalition parties in the successful control action satisfies the requirements and forms a solution to the $k$-subset sum problem.

\end{proof}

\subsection{Proof of Theorem \ref{thm:dop-general-C}}
The proof is by a reduction from the k-clique problem.
\begin{reminder}{\ref{thm:dop-general-C}}
    For general preferences, control for the coalition by deleting opposition parties (\withoutTwithoutLdop) is W[1]-hard.
\end{reminder}
\begin{proof}
    Given an instance of the k-clique problem. Let the graph be $G=(W, D)$, where $|W| = n$, $|E| = m$, and parameter $k$, where the parameter is the size of the clique. 
    Construct an instance of the \withoutTwithoutLdop\ problem with parameter $k$ such that there is a successful control action if and only if a k-clique exists. 

    Let $P=\{p_1,s_1,\ldots, s_n\}$, each $s_i$ is corresponding to a vertex in $G$.
    Let the coalition to be $C=\{p_1\}$. 
    Let the parameter of the number of deleting parties be $k$.
    Let $\varphi = \frac{k^2}{m}$ (note that, if $\varphi>1$, then there is no clique of size $k$ in the graph since there are no $k^2$ edges at all).

    \textbf{Voter Preference Order:}
    Let $V = \{v_1,\ldots v_m\}$  be the set of voters, each voter is corresponding to one edge in $G$. 
    The  preference order of $v_i$, that corresponding to the edge $(w_x,w_y)$ is defined to be $s_x \succ s_y \succ p_1 \succ P\setminus \{s_x,s_y,p_1\}$.

    \textbf{To Sum Up.}
    Assume there exists a k-clique in the graph, then delete the corresponding $k$ parties.
    Each edge between these vertices is in the graph, hence there are exactly $k^2$ such edges. Thus, in the preference order of all voters that correspond to these edges, after the deletion, the first party will be $p_1$, and the coalition achieve $\frac{k^2}{m}$ as desired.

    On the other hand, suppose a successful control action exists.  
    In the original election, $p_1$ does not receive any votes, and there are $m$ voters participating.  
    Thus, the control action must secure at least $k^2$ votes by deleting $k$ parties.  

    If a deletion results in a vote for $p_1$, it implies that both vertices corresponding to the endpoints of the associated edge are among the deleted parties.  
    This means the deletion of $k$ parties successfully includes the vertices representing both endpoints of the edges for $k^2$ edges.  

    Consequently, all edges between these vertices must exist in the graph $G$, indicating that these vertices form a clique. 
\end{proof}

\subsection{Proof of Theorem \ref{thm:dop-ssp-CFP-non-con}}
The proof is by a reduction from the k-subset sum problem.
\begin{reminder}{\ref{thm:dop-ssp-CFP-non-con}}
   For SSP preferences, control for the coalition and favored party by deleting opposition parties (\withoutTwithLdop) is NP-hard under the compact representation.
\end{reminder}
\begin{proof}
    Given an instance of the $k$-subset sum problem, let $A = \{a_1, a_2, \ldots, a_m\}$ be the set of integers, $\tau$ be the target integer, and $k$ be the parameter. We construct an instance of the \withoutTwithLdop\ problem such that there is a solution to the $k$-subset sum problem if and only if there exists a successful control action.

    Let the set of parties be $P = \{p_1, b_1, b_2, \ldots, b_m, s_1, s_2, \ldots, s_m\}$.
    The coalition is $C = \{p_1, s_1, s_2, \ldots, s_m\}$ and the favored party is $p_1$.
    
    That is, each integer $a_i$, has a corresponding opposition party $b_i$, and a corresponding coalition coalition party $s_i$.

    Place the parties on the interval such that the closest party to each $s_i$ is $b_i$.

    Let the set of voters $V$ include:
    \begin{itemize}
        \item For each $a_i \in A$, voters $v_{i,1}, v_{i,2}, \ldots, v_{i,a_i}$ whose peak is at the opposition party $b_i$, that is $h(v_{i,j}) = b_i$.
        \item $\tau$ voters $v_1, v_2, \ldots, v_{\tau}$ whose peak is at the favored party $p_1$, that is $h(v_i) = p_1$.
    \end{itemize}

    Define the parameter for the number of added parties as $k$.
    Let the target fraction be $\varphi = \frac{2 \cdot \tau}{\tau + \sum_{a \in A} a}$ and the target ratio be $\rho = \frac{1}{2}$.

    \paragraph{Summary.}
    All voters of a deleted opposition party move to the corresponding coalition party. Thus, before deleting any opposition party, the opposition receives $\sum_{a \in A} a$ votes, and the favored party receives $\tau$ votes.

    Assume there is a subset $A' \subseteq A$ such that $|A'| \leq k$ and $\sum_{a \in A'} a = \tau$. deleting the corresponding opposition parties results in exactly $\tau$ new votes for the coalition, satisfying both requirements.

    On the other hand, assume there is a successful control action. This implies there is a set of at most $k$ opposition parties whose deletion meets both requirements.

    (1) Since the requirement of $\varphi$ holds, the coalition must receive at least $2 \cdot \tau$ votes from the total $\tau + \sum_{a \in A} a $ votes.

    (2)Since the favored party receives only $\tau$ votes and cannot gain additional votes through control, the other coalition members can collectively achieve at most $\tau$ votes.

    Therefore, the coalition receives exactly $\tau$ votes.

    Deleting any opposition party $b_i$ contributes exactly $a_i$ votes to the coalition. Hence, the subset containing elements corresponding to the deleted opposition parties in control satisfies the requirements and forms a solution to the $k$-subset sum problem.
\end{proof}

\end{document}